%% file: main.tex
\documentclass[journal]{IEEEtran}

\usepackage{amsmath,amssymb,amsfonts}
\usepackage{amsthm}
\usepackage{graphicx}
\usepackage{acronym}
\usepackage{siunitx}
\usepackage{subcaption}
\usepackage{cite}
\usepackage{stfloats}
\usepackage{booktabs}
\usepackage{url}
\usepackage{mathtools}
\usepackage{multirow}
\usepackage[table]{xcolor}

\newtheorem{lemma}{Lemma}
\newtheorem{theorem}{Theorem}
\newtheorem{proposition}{Proposition}

\input{acronyms.tex}

\begin{document}

\title{Analytical Modeling of Open- and Closed-Loop Dispersive Molecular Communication Channels with Pulsatile Flow}

\author{Theofilos Symeonidis,
Fardad Vakilipoor,
Timo Jakumeit,
Robert Schober,\\
Nunzio Tuccitto,
and Maximilian Sch\"afer%
\thanks{This paper was presented in part at the ACM International Conference on Nanoscale Computing and Communication, 2026~\cite{symeonidis2026analyticalmodelingdispersiveclosedloop}, doi: 10.1145/3818305.3830261.}
\thanks{T. Symeonidis, F. Vakilipoor, T. Jakumeit, R. Schober, and M. Sch\"afer are with the Institute for Digital Communications, Friedrich-Alexander-Universit\"at Erlangen-N\"urnberg, Erlangen, Germany.}
\thanks{N. Tuccitto is with the University of Catania, Catania, Italy.}
}

\markboth{IEEE Journal Paper,~2026}%
{Symeonidis \MakeLowercase{\emph{et al.}}}

\maketitle

\begin{abstract}
\Ac{MC} is a communication paradigm in which information is conveyed through
the release, propagation, and reception of molecules. Many envisioned healthcare applications of \ac{MC} are expected to operate
inside the human body, where the \ac{CVS} may serve as the physical
propagation environment and molecular transport is governed by diffusion
and blood flow. Although blood flow
is inherently pulsatile, most analytical \ac{MC} channel models assume steady
flow. In this paper, we develop a time-variant analytical model for dispersive
\ac{MC} channels with pulsatile flow. We derive the straight-duct response as a Normal distribution with
time-variant mean and variance, capturing the combined effects of diffusion
and pulsatile flow, and extend it to closed-loop channels through a
wrapped-Normal representation. The model is
validated against \ac{3D} \acp{PBS} for synthetic and physiologically
motivated velocity waveforms. We further derive a closed-form first-order
approximation for the first-arrival peak time and introduce the
nondimensional indicator \(S_{\mathrm{Rx}}\) for assessing the applicability
of the reduced \ac{1D} model.

Our results show that pulsatility has the strongest influence when molecular
transport is predominantly advective and the temporal flow variations are not
averaged out, whereas stronger diffusion and faster pulsations reduce its
impact on the received signal. Moreover, a \ac{PBS}-based parameter sweep
supports \(S_{\mathrm{Rx}}\ge2\) as a practical criterion for the applicability of the proposed analytical model. Finally, relating the analytical assumptions
to representative blood-vessel classes shows that model applicability depends
not only on the transport conditions but also on physiological properties such as
vessel rigidity, geometric uniformity, and blood rheology.
\end{abstract}


\IEEEpeerreviewmaketitle
\acresetall

\section{Introduction}
\Ac{MC} is an emerging interdisciplinary research field that facilitates information transfer in biological environments. \Ac{MC} has a wide range of potential biomedical and healthcare applications, including health monitoring, early disease detection, and targeted treatment~\cite{6548006,8710366,8567997}.
Most of these applications are expected to operate inside the human body, specifically in the \ac{CVS}, which constitutes the physical propagation channel, where signaling molecules are transported assisted by blood flow~\cite{FELICETTI201627,10433775}. Therefore, the successful design of in-body \ac{MC} systems requires meaningful models that account for all relevant effects arising in this environment.

While the above models provide important insight into flow-based molecular
transport, practical in-body \ac{MC} systems require descriptions that are
both locally accurate within individual vessel segments and compatible with
the continuously circulating nature of blood flow. Locally, molecule transport is governed by the combined effects of diffusion
and flow, which determine the received signal for a given \ac{Tx}-\ac{Rx} pair within a vessel segment. At
longer time scales, circulation through the vascular network may introduce
long-term memory, repeated molecule arrivals, and coupling between local
propagation dynamics and the overall network structure
\cite{brand2025closedloop}.

Recent studies have started to investigate \ac{MC} in closed-loop
systems~\cite{ieee8618321,brand2025closedloop}. These studies demonstrate
that recirculation can lead to channel characteristics that differ
substantially from those of non-recirculating systems, including extended
delay spreads, persistent inter-symbol interference, and topology-dependent
received signals. Experimental studies further confirm the presence of these
effects in closed-loop testbeds~\cite{vakilipoor2025cam,11184881},
while the results in~\cite{vakilipoor2025cam} also show that complex vascular
networks can be well approximated by an equivalent single closed-loop topology,
thereby enabling a simplified representation of recirculating molecular transport. Together, these observations motivate the development
of channel models that retain the local propagation behavior of vessel
segments while being extensible to closed-loop topologies. The equivalent-loop representation, however, is a network-level abstraction
and does not imply uniform local transport conditions throughout the vascular
network. Individual vessel segments may differ substantially in size and flow
pulsatility. Therefore, in addition to the closed-loop representation, we also
consider the straight-duct response as a model for an individual link in a vascular network, where local pulsatile-flow effects can be studied
independently of recirculation.

Another important feature of in-body transport is that blood flow is
generally time-varying. In particular, flow in the \ac{CVS} is inherently
pulsatile due to the cardiac cycle. This physiological feature is well
established in the hemodynamics and fluid mechanics literature
\cite{https://doi.org/10.1002/j.2040-4603.2016.tb00700.x,https://doi.org/10.1002/smll.201904032}
and is further supported by in-vivo measurements and numerical studies that
provide realistic time-variant blood-flow shapes
\cite{https://doi.org/10.1002/fld.1966,Holdsworth_1999}. Nevertheless, most
analytical \ac{MC} channel models for flow-driven or closed-loop transport assume steady flow conditions~\cite{ieee8618321,brand2025closedloop,vakilipoor2025cam,11184881}.

A small number of \ac{MC} studies have adopted pulsatile flow models and
demonstrated their influence on molecule transport and channel responses.
In~\cite{6548006}, an analytical network-based model for drug-particle
propagation in the \ac{CVS} under pulsatile blood flow is developed,
combining a cardiovascular network model with a drug propagation network
model to characterize the received particle concentration at a target site. However, the resulting network-based formulation is highly complex and does
not yield a compact closed-form \ac{CIR}, limiting its tractability for
communication-theoretic analysis. In~\cite{11157771},
nanoparticle propagation with pulsatile flow is investigated using
\acp{PBS}, revealing significant deviations
compared to the steady-flow case. However, the study does not provide an
analytical channel model, neglects diffusive transport, and is tailored to a
specific physiological scenario, limiting its general applicability.
In~\cite{11161338}, analytical \acp{CIR} for non-Newtonian flow in a
plaque-obstructed vessel are derived and complemented by detailed computational fluid dynamics
simulations incorporating pulsatile flow effects. However, the study relies on numerical simulations, while the analytical models are shown to provide only
rough approximations of the actual flow behavior and do not constitute a unified
analytical description of pulsatile \ac{MC} channels.
In~\cite{lee2026correctedinversegaussianfirsthittingtimemodelingmolecular},
an analytical channel model for the first-hitting-time under time-variant drift
is derived. However, the model is restricted to one-dimensional propagation
and focuses on first-hitting-time statistics for absorbing \acp{Rx}. Moreover, pulsatile flow is represented via a
spatially uniform time-variant drift and approximations are used to obtain a
tractable expression, limiting its applicability to realistic flow conditions.
In addition, none of the above works provides an analytical model that
combines pulsatile-flow-induced dispersion with a closed-loop channel topology.

Motivated by these observations, we propose a time-variant analytical
\ac{MC} channel model for pulsatile flow in the dispersive regime~\cite{Jamali2019ChannelModeling}. The model
builds on recent results in~\cite{wang2025pulsatile}, which show that
pulsatile flow described by a Womersley velocity profile~\cite{womersley1955} can be reduced to an
effective \ac{1D} advection-diffusion equation with time-variant velocity and
diffusion coefficient under suitable model assumptions. Based on this description, we derive the \ac{CIR} for a straight \ac{MC} channel with pulsatile flow, and then extend it to closed-loop systems. Unlike existing models, the resulting analytical model can both describe
the received signal within an individual vessel segment and account for the possible
recirculation of molecules in closed-loop \ac{MC} systems. The main contributions of this paper can be summarized as follows:

\begin{itemize}
\item We develop a time-variant dispersive \ac{1D} transport model for molecule propagation with pulsatile flow by combining Womersley flow with
Aris-Taylor dispersion. The resulting model captures both the
time-variant mean displacement and the effective dispersion
induced by the interplay of diffusion and pulsatile flow.

\item We derive the \ac{CIR} for impulsive molecule release in a straight duct
with diffusion and pulsatile flow. The resulting expression is a Normal distribution with
time-variant mean and variance, which are obtained in closed form based on the effective time-variant velocity and diffusion coefficient,
respectively.

\item We extend the straight-duct solution to closed-loop \ac{MC} channels by imposing
periodic boundary conditions, yielding an analytical \ac{CIR} for \ac{MC} systems with diffusion and pulsatile flow that accounts for
molecule recirculation.

\item We investigate the influence of pulsation on the first-arrival peak time, and derive a closed-form first-order approximation for the direct
first-arrival peak time and evaluate its accuracy.

\item We validate the proposed analytical models by comparison to \ac{3D} \acp{PBS} for
synthetic and physiologically motivated velocity waveforms. Moreover, we compare the results of the proposed time-variant models with corresponding steady-flow models to investigate the effect of pulsation on the received signal.

\item We analyze the model applicability over a broad parameter range and
introduce an \ac{Rx}-distance-based nondimensional indicator
\(S_{\mathrm{Rx}}\), which summarizes the conditions under which the proposed model is expected to be accurate. Based on this indicator, we discuss
the applicability of the proposed model to representative blood-vessel classes.

\end{itemize}

This paper substantially extends the preliminary conference version
\cite{symeonidis2026analyticalmodelingdispersiveclosedloop} in several
directions. Building on the straight-duct derivation and closed-loop
expression introduced there, this paper provides more detailed derivations, and
places greater emphasis on the local straight-duct \ac{Rx} response. We further derive a closed-form
first-order approximation for the first-arrival peak time for \ac{MC} systems with pulsatile flow. In addition, we
introduce a regime analysis based on a nondimensional admissibility indicator
and relate the resulting model assumptions to representative blood-vessel
classes.

The remainder of this paper is organized as follows. In
Section~\ref{sec:model}, we introduce the physical channel model, starting
from the underlying \ac{3D} transport formulation, and derive the effective
\ac{1D} dispersive model. Section~\ref{sec:channel}
derives the corresponding \ac{MC} channel model, first for a straight duct
and then for a closed-loop channel, and develops a first-order approximation
for the first-arrival peak time. Section~\ref{sec:results} introduces the
considered velocity waveforms and the \ac{3D} \ac{PBS} setup, validates
the direct \ac{Rx} response for synthetic and physiologically motivated
pulsatile flow scenarios, and evaluates the accuracy of the first-arrival
peak-time approximation. Section~\ref{sec:regime_analysis} analyzes the
parameter regimes in which the proposed model is applicable, introduces the
nondimensional indicator \(S_{\mathrm{Rx}}\), and discusses the implications
of the model assumptions for representative blood-vessel classes. Finally,
Section~\ref{sec:conclusion} concludes the paper.

\section{Physical Channel Model}\label{sec:model}

In this section, we introduce the physical channel model for pulsatile
molecule transport in a cylindrical vessel segment and subsequently extend it
to closed-loop \ac{MC} channels. Motivated by the results in
\cite{vakilipoor2025cam}, which show that complex vascular networks can be
approximated by an equivalent single closed-loop topology, we use this
single-loop representation to account for molecular recirculation. The
resulting system model is illustrated in Fig.~\ref{fig:system_model}.
Starting from the \ac{3D} advection-diffusion equation with space- and  time-variant flow in cylindrical coordinates, as introduced by Aris~\cite{aris1956}, we obtain a tractable description of the axial transport by deriving a reduced \ac{1D} dispersive transport model~\cite{wang2025pulsatile}. We then adopt a Womersley flow model for the pulsatile velocity field and define the corresponding effective diffusion coefficient.

\begin{figure*}[t]
  \centering
  \includegraphics[width=0.9\textwidth]{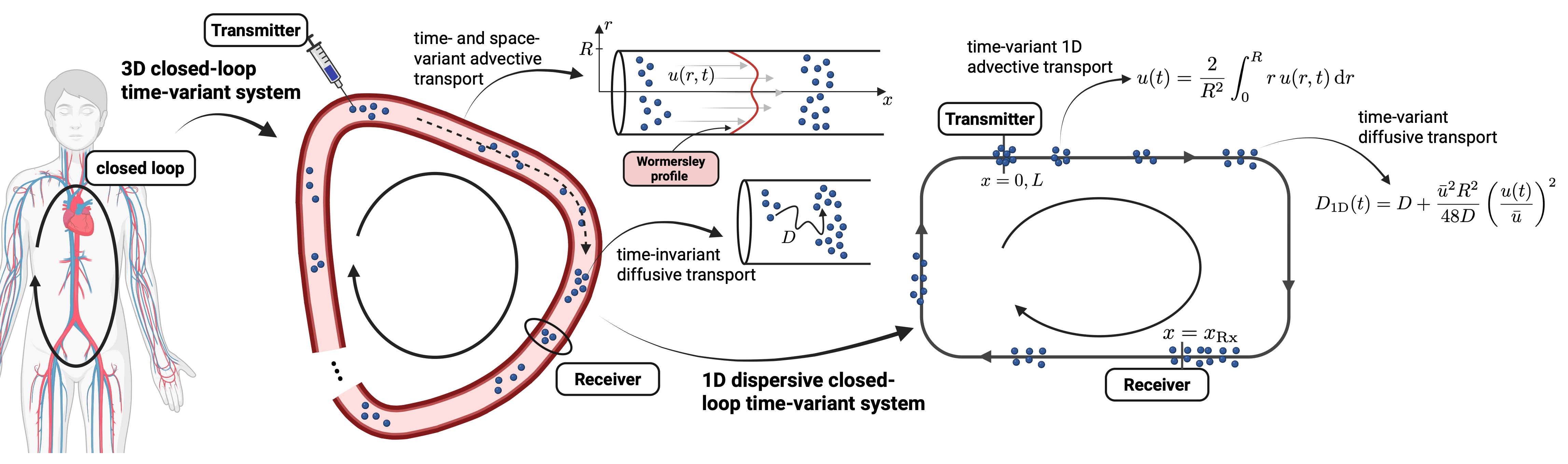}
  \captionsetup{font=normalsize}
  \caption{\textbf{Closed-loop molecular communication system and model reduction.}
  Signaling molecules are released by the \ac{Tx} and propagate through a
  vascular-like closed-loop channel under pulsatile flow (left). The resulting
  transport is governed by a time- and space-variant advection-diffusion process
  with Womersley velocity profile (center). In the dispersive regime, the system
  can be approximated by a \ac{1D} model with time-variant effective velocity
  and diffusion coefficient (right).}
  \label{fig:system_model}
\end{figure*}

\subsection{\ac{3D} Transport Model}
Blood-vessel segments are commonly approximated as cylindrical channels for
analytical modeling of blood flow and solute transport~\cite{womersley1955,wang2025pulsatile}.
Accordingly, in a cylindrical channel, the transport of signaling molecules is
governed by \ac{3D} advection and diffusion. Following the classical formulation by Aris~\cite{aris1956}, the normalized particle concentration $p(x,r,t)$, where $x$ and $r$ denote the axial and radial coordinates, respectively, and $t$ denotes time, can be described by the radially symmetric advection-diffusion equation~\cite{wang2025pulsatile}
\begin{align}
    \begin{split}
\partial_t p(x,r,t)+&u(r,t)\partial_x p(x,r,t)
\!=\\&D\!\left(\partial_{xx}p(x,r,t)\!+\!\frac{1}{r}\partial_r\!\big(r\,\partial_r p(x,r,t)\big)\right),
\label{eq:3D_adv_diff}
    \end{split}
\end{align}
where $\partial_t$, $\partial_x$, and $\partial_r$ denote the first-order
derivatives with respect to time $t$, axial coordinate $x$, and radial
coordinate $r$, respectively, and $\partial_{xx}$ denotes the second-order
derivative with respect to the axial coordinate $x$. Moreover, $u(r,t)$ is the axial \ac{3D} pulsatile flow velocity profile with pulsation period $T$, and $D$ is the molecular diffusion coefficient.

\subsection{Dispersive Regime and \ac{1D} Model Reduction}
\label{sec:assumptions}

Following the Aris-Taylor dispersion theory and its extension to time-variant flows, the reduction of \eqref{eq:3D_adv_diff} to a \ac{1D} equation relies on the following assumptions~\cite{wang2025pulsatile,aris1956}:

\begin{itemize}
  \item The channel is long and slender, i.e.,
\begin{equation}
R \ll L,
\label{eq:slender}
\end{equation}
where $R$ and $L$ denote the channel radius and length, respectively.
  
  \item Radial diffusion is much faster than axial transport, such that the concentration becomes approximately uniform over the cross section, which requires
\begin{equation}
\frac{R^2}{D} \ll \frac{L}{\bar{u}},
\label{eq:radial_fast}
\end{equation}
where $\bar{u}$ is the temporal mean of the cross-sectionally averaged axial flow velocity over one pulsation period $T$, i.e.,
\begin{equation}
\bar{u}
=
\frac{2}{R^{2}T}
\int_{0}^{T}
\int_{0}^{R}
r\,u(r,t)\,\mathrm{d}r\,\mathrm{d}t .
\label{eq:u_bar_def}
\end{equation}

\item Axial dispersion is dominated by shear-induced spreading, which arises from the combined effect of transverse diffusion and the non-uniform flow profile across the channel cross section. In particular, molecular diffusion in the axial direction alone is negligible, i.e.,
\begin{equation}
\frac{L}{\bar{u}} \ll \frac{L^2}{D}.
\label{eq:axial_dispersion_condition}
\end{equation}
  
\item The flow is laminar, fully developed, and radially symmetric, and the carrier fluid is Newtonian.
\end{itemize}

Under these conditions, solute transport enters the dispersive regime, in which the concentration becomes rapidly homogenized over the cross-section, and the longitudinal dynamics can be described by an effective \ac{1D} advection-diffusion equation with time-variant coefficients. Accordingly, the normalized concentration in \eqref{eq:3D_adv_diff} can be averaged over the cross-section as follows
\begin{equation}
p(x,t) = \frac{2}{R^2}\int_0^R p(x,r,t)\, r\,\mathrm{d}r.
\end{equation}
For the case of Womersley velocity profiles, i.e., pulsatile flow profiles in cylindrical channels described by the Womersley solution~\cite{womersley1955}, \eqref{eq:3D_adv_diff} can be transformed into an effective \ac{1D} advection-diffusion equation~\cite[Eq.~(41)]{wang2025pulsatile}
\begin{equation}
\partial_t p(x,t)
+ u(t)\partial_x p(x,t) =
D_{\mathrm{1D}}(t) \partial_{xx} p(x,t).
\label{eq:adv_diff}
\end{equation}
Here, $u(t)$ is the cross-sectionally averaged axial flow velocity obtained from the underlying \ac{3D} velocity profile $u(r,t)$, and $D_{\mathrm{1D}}(t)$ is the effective time-variant diffusion coefficient that accounts for the combined effect of molecular diffusion and shear-induced dispersion. In the following, we characterize the effective time-variant velocity $u(t)$ and diffusion coefficient $D_{\mathrm{1D}}(t)$ in detail.

\subsection{Womersley Velocity Profile}
\label{subsec:womersley_flow}

The Womersley velocity profile is commonly used to model pulsatile flow in cylindrical vessels~\cite{womersley1955,11157771,wang2025pulsatile}. Here, we adopt this model for the \ac{3D} axial velocity profile as given in~\cite[Eq.~(31)--(39)]{wang2025pulsatile}. The axial velocity $u(r,t)$ in \eqref{eq:3D_adv_diff} is modeled as the superposition of a steady Poiseuille component and a pulsatile component composed of harmonic Womersley modes, i.e.,
\begin{equation}
u(r,t)=u_{\mathrm{P}}(r)+\Re\!\left\{u_{\mathrm{W}}(r,t)\right\},
\label{eq:u_3d_decomp}
\end{equation}
with the steady contribution $u_{\mathrm{P}}(r)=2\bar{u}\left(1-\frac{r^2}{R^2}\right)$, where $\bar{u}$ is the cross-sectionally averaged axial velocity~\eqref{eq:u_bar_def}, and $\Re\{\cdot\}$ the real-part operator.

The pulsatile Womersley component \(u_{\mathrm W}(r,t)\) in
\eqref{eq:u_3d_decomp} is defined as
\begin{equation}
u_{\mathrm{W}}(r,t)
=
\sum_{n=1}^{N}
U_w(\omega_n)\,
\Psi_n(r)\,
\exp\!\left(\mathrm{j}(n\omega t+\phi_n)\right),
\label{eq:u_womersley}
\end{equation}
where $U_w(\omega_n)$ denotes the amplitude of the $n$-th harmonic with angular frequency $\omega_n = n\omega$, $\phi_n$ is the corresponding phase shift, $N$ is the number of harmonics, and $\mathrm{j}$ is the imaginary unit. The fundamental angular frequency is given by $\omega=2\pi f$, where $f=1/T$ is the pulsation frequency.
The radial dependence of each harmonic in \eqref{eq:u_womersley} is described by the Womersley shape function
\begin{equation}
\Psi_n(r)
=
\frac{
J_0\!\left(\mathrm{j}^{3/2}a_n\right)
-
J_0\!\left(\mathrm{j}^{3/2}a_n \frac{r}{R}\right)
}{
J_0\!\left(\mathrm{j}^{3/2}a_n\right)
-
\frac{2J_1\!\left(\mathrm{j}^{3/2}a_n\right)}{\mathrm{j}^{3/2}a_n}
},
\label{eq:womersley_shape}
\end{equation}
where $J_0(\cdot)$ and $J_1(\cdot)$ denote the Bessel functions of the first kind of order zero and one, respectively. The corresponding Womersley number of the $n$-th harmonic is given by
\begin{equation}
a_n = R\sqrt{\frac{n\omega}{\nu}},
\label{eq:womersley_number}
\end{equation}
where $\nu$ in \eqref{eq:womersley_number} is the kinematic viscosity of the fluid. The Womersley number in \eqref{eq:womersley_number} quantifies the relative importance of unsteady inertial effects and viscous diffusion, and indicates how rapidly the velocity profile adapts to temporal variations of the flow~\cite{womersley1955}. Small Womersley numbers correspond to quasi-steady flow conditions where viscous effects dominate and the velocity profile remains close to parabolic, whereas large Womersley numbers indicate inertia-dominated flow with a flattened velocity profile and significant phase lag between the pressure gradient and the flow.

From \eqref{eq:u_3d_decomp}, the effective \ac{1D} velocity $u(t)$ in
\eqref{eq:adv_diff} is obtained by cross-sectional averaging as follows
\begin{align}
    \begin{split}
u(t)&=\frac{2}{R^2}\int_0^R r\,u(r,t)\,\mathrm{d}r
\\&=\bar{u}\!\left(
1 + \sum_{n=1}^{N} M_n \cos(n\omega t + \phi_n)
\right),
\label{eq:u_t}
    \end{split}
\end{align}
where the dimensionless coefficients $M_n$ are directly related to the amplitudes of the corresponding \ac{3D} Womersley harmonics, i.e., $M_n=\frac{U_w(\omega_n)}{\bar{u}}$.

\subsection{Time-Variant Diffusion Coefficient}
Similar to the effective diffusion coefficient for steady-flow systems~\cite{Jamali2019ChannelModeling}, the time-variant effective diffusion coefficient in \eqref{eq:adv_diff} is determined by the pulsatile flow velocity described by the Womersley model in Section~\ref{subsec:womersley_flow}. For the parameter ranges considered in this paper, the corresponding Womersley numbers $a_n$ are within the regime for which the diffusion coefficient approximation derived in~\cite[Eq.~(43)]{wang2025pulsatile} can be directly applied. Then, the time-variant effective diffusion coefficient in \eqref{eq:adv_diff} follows as
\begin{equation}
D_{\mathrm{1D}}(t)
=
D + \frac{\bar{u}^2 R^2}{48D}
\left(\frac{u(t)}{\bar{u}}\right)^2.
\label{eq:D1D}
\end{equation}

\section{\ac{MC} Channel Model}\label{sec:channel}
In this section, we derive the \ac{MC} channel model by solving the \ac{1D}
advection-diffusion equation in \eqref{eq:adv_diff}. We first consider a
straight duct without recirculation, representing the local transport behavior
of an individual vessel segment within a larger vascular network, and obtain
the corresponding solution via a sequence of variable transformations. We then
extend this solution to closed-loop \ac{MC} channels.

\subsection{Straight Duct Channel Model}
\label{subsec:str_duct}

\begin{lemma}
\label{lem:straight_duct}
For an impulsive release of particles at $t=0$ and $x=0$, the solution to \eqref{eq:adv_diff} in an infinitely long straight duct can be obtained as follows
\begin{equation}
p_{\mathrm{sd}}(x,t)
=
\frac{1}{\sqrt{2\pi\sigma^2(t)}}
\exp\!\left(
-\frac{[x-\mu(t)]^2}{2\sigma^2(t)}
\right),
\label{eq:gaussian_infinite}
\end{equation}
with the time-variant mean $\mu(t)$ and variance $\sigma^2(t)$ given by
\begin{align}
&\mu(t) = \int_0^t u(s)\,\mathrm{d}s,
&&\sigma^2(t) = 2\int_0^t D_{\mathrm{1D}}(s)\,\mathrm{d}s.
\label{eq:mean_var_def}
\end{align}
Inserting the pulsatile velocity model \eqref{eq:u_t} into \eqref{eq:mean_var_def}, an analytical expression for the mean $\mu(t)$ follows as
\begin{equation}
\mu(t)
=
\bar u\left[
t
+
\sum_{n=1}^{N}
\frac{M_n}{n\omega}
\big(
\sin(n\omega t+\phi_n)-\sin\phi_n
\big)
\right].
\label{eq:mu_closed}
\end{equation}
Moreover, substituting \eqref{eq:u_t} and \eqref{eq:D1D} into \eqref{eq:mean_var_def} yields an analytical expression for the time-variant variance $\sigma^2(t)$, i.e.,
\begin{equation}
\begin{aligned}
\sigma^2(t)
=
2Dt
+
\frac{R^2\bar u^2}{24D}
\Bigg[
& t
+2\sum_{n=1}^{N} M_n S_n(t)
+\sum_{n=1}^{N} M_n^2 Q_n(t) \\
& +2\sum_{1\le m<n\le N} M_m M_n P_{m,n}(t)
\Bigg],
\end{aligned}
\label{eq:sigma2_closed}
\end{equation}
with the coefficients, defined for $n=1,\ldots,N$ and
$1\le m<n\le N$,
\begin{align}
S_n(t)
&=
\frac{\sin(n\omega t+\phi_n)-\sin(\phi_n)}{n\omega},
\\[0.6ex]
Q_n(t)
&=
\frac{t}{2}
+
\frac{\sin(2n\omega t+2\phi_n)-\sin(2\phi_n)}{4n\omega},
\\[0.6ex]
P_{m,n}(t)
&=
\frac{1}{2}\Bigg[
\frac{\sin((n-m)\omega t+(\phi_n-\phi_m))
      -\sin(\phi_n-\phi_m)}{(n-m)\omega}
\nonumber\\
&\hspace{0.2cm}
+
\frac{\sin((n+m)\omega t+(\phi_n+\phi_m))
      -\sin(\phi_n+\phi_m)}{(n+m)\omega}
\Bigg].
\end{align}
\end{lemma}

\noindent\textit{Proof:}
The proof is provided in Appendix~\ref{app:straight_duct_proof}.

\subsection{Extension to Closed-Loop Channels}

To account for the recirculations in closed-loop \ac{MC} systems (see Fig.~\ref{fig:system_model}), we extend the solution for the infinitely long straight duct in \eqref{eq:gaussian_infinite} in the following.

\begin{theorem}
\label{thm:closed_loop}
For a closed-loop channel of length $L$ with pulsatile flow, the concentration
$p(x,t)$ at position $x\in[0,L)$ can be represented by the following wrapped
normal distribution:
\begin{equation}
p(x,t)
=
\sum_{k=-\infty}^{\infty}
\frac{1}{\sqrt{2\pi\sigma^2(t)}}
\exp\!\left(
-\frac{[x-\mu(t)+kL]^2}{2\sigma^2(t)}
\right),
\label{eq:cir}
\end{equation}
where $k\in\mathbb{Z}$ indexes the periodic images of the concentration
distribution. In particular, the term with $k=0$ corresponds to the direct
contribution without a complete circulation. For the considered positive
mean-flow direction, the terms with $k>0$ account for contributions after
one or more complete forward circulations, whereas the terms with $k<0$
account for periodic-image contributions associated with molecular spreading
in the opposite axial direction across the periodic boundary. Equation
\eqref{eq:cir} thus extends the straight-duct result in
\eqref{eq:gaussian_infinite} to a recirculating channel while preserving the
time-variant mean $\mu(t)$ and variance $\sigma^2(t)$ derived in
\eqref{eq:mean_var_def}.
\end{theorem}

\noindent\textit{Proof:}
The proof is provided in Appendix~\ref{app:closed_loop_proof}.

\vspace{0.5em}
\noindent For detection, we consider a transparent \ac{Rx} of length
$\Delta x_{\mathrm{Rx}}$ located at position $x=x_{\mathrm{Rx}}$
(see Fig.~\ref{fig:system_model}). The received signal
$p_{\mathrm{Rx}}(t)$ is defined as the fraction of particles within the
\ac{Rx} region and is obtained by integrating the normalized spatial
concentration over the \ac{Rx} region as follows
\begin{equation}
p_{\mathrm{Rx}}(t)
=
\int_{x_{\mathrm{Rx}}-\Delta x_{\mathrm{Rx}}/2}^{x_{\mathrm{Rx}}+\Delta x_{\mathrm{Rx}}/2}
p(x,t)\,\mathrm{d}x.
\label{eq:pRx}
\end{equation}

\subsection{High-Frequency Limit of Pulsatile Flow}
\label{subsec:high_frequency_limit}

In this subsection, we analyze the asymptotic behavior of the proposed channel model in the high-frequency regime, i.e., for $\omega \to \infty$. This analysis provides insight into the temporal averaging effect observed for rapidly varying flow conditions.
We note that this limit is understood as a formal high-frequency limit of the
reduced one-dimensional model. Since the Womersley number $a_n$, defined in
\eqref{eq:womersley_number}, scales as $a_n\propto\sqrt{n\omega}$,
arbitrarily large values of $\omega$ would eventually lead to a violation of the underlying assumptions.

From \eqref{eq:mu_closed}, the mean displacement consists of the steady
contribution \(\bar{u}t\) and oscillatory contributions that scale with
\(1/(n\omega)\). Since the sinusoidal terms are bounded, these oscillatory
contributions vanish as \(\omega\to\infty\). Therefore,
\begin{equation}
\lim_{\omega \to \infty} \mu(t) = \bar{u} t,
\label{eq:omega_inf}
\end{equation}
which coincides with the mean of the corresponding steady-flow
channel. This result shows that, in the high-frequency regime, the time-variant flow effectively averages out, and the mean particle transport is governed solely by the temporal mean velocity $\bar{u}$.

Next, we consider the variance $\sigma^2(t)$ in \eqref{eq:sigma2_closed}. By analyzing the individual terms, it can be shown that all oscillatory contributions vanish as $\omega\to\infty$, while the quadratic terms in the velocity amplitudes remain. Consequently, we obtain
\begin{equation}
\lim_{\omega \to \infty} \sigma^2(t)
=
2 D_{\mathrm{eff}}^\infty t,
\end{equation}
where the effective diffusion coefficient is given by
\begin{equation}
D_{\mathrm{eff}}^\infty
=
D + \frac{R^2 \bar{u}^2}{48D}
\left(
1 + \frac{1}{2}\sum_{n=1}^{N} M_n^2
\right).
\end{equation}

For comparison, the effective diffusion coefficient of the corresponding steady-flow system is~\cite{Jamali2019ChannelModeling}
\begin{equation}
D_{\mathrm{eff}}^{\mathrm{const}}
=
D + \frac{R^2 \bar{u}^2}{48D}.
\label{eq:Deff_const}
\end{equation}

Thus, the difference between the effective diffusion coefficient of the
pulsatile system in the high-frequency limit and that of the corresponding
steady-flow system is given by
\begin{equation}
\Delta D_{\mathrm{eff}}
=
D_{\mathrm{eff}}^\infty - D_{\mathrm{eff}}^{\mathrm{const}}
=
\frac{R^2 \bar{u}^2}{96D}
\sum_{n=1}^{N} M_n^2,
\label{eq:delta_Deff}
\end{equation}
which is strictly positive and shows that pulsation increases the effective
dispersion relative to steady flow. More generally, this increase is not
restricted to the high-frequency limit, since pulsation also increases the
cycle-averaged effective dispersion for a given mean flow velocity.

This result shows that, within the reduced model, the mean behavior of a
pulsatile system converges to that of a steady-flow system, with larger
effective dispersion. The additional dispersion depends on the squared
amplitudes of the harmonic components. The implications of this asymptotic
behavior for the received signal are investigated in
Section~\ref{sec:results}.

\subsection{First-Arrival Peak-Time Approximation}
\label{subsec:first_peak_analytical_approx}

Pulsatile flow can influence the temporal shape of the \ac{Rx} signal. In
particular, pulsatility can shift the first-arrival peak time, which directly
affects the choice of \ac{Rx} sampling times and is therefore relevant for
\ac{Rx} design. In the following, we derive a first-order approximation for the
first-arrival peak time based on the analytical channel model derived above.

We define the first-arrival peak time as the time at which the finite-region
\ac{Rx} response obtained by inserting the straight-duct solution
\(p_{\mathrm{sd}}(x,t)\) in \eqref{eq:gaussian_infinite} into
\eqref{eq:pRx} reaches its maximum, i.e.,
\begin{equation}
t_{\mathrm{peak}}
=
\arg\max_{t>0}
\int_{x_{\mathrm{Rx}}-\Delta x_{\mathrm{Rx}}/2}^{x_{\mathrm{Rx}}+\Delta x_{\mathrm{Rx}}/2}
p_{\mathrm{sd}}(x,t)\,\mathrm{d}x.
\label{eq:t_peak_definition}
\end{equation}
For the closed-loop channel, this response corresponds to the \(k=0\) contribution of the
wrapped-Normal representation in \eqref{eq:cir}, since setting \(k=0\) in
\eqref{eq:cir} directly recovers the straight-duct solution
\(p_{\mathrm{sd}}(x,t)\) in \eqref{eq:gaussian_infinite}.
Equation~\eqref{eq:t_peak_definition} defines the finite-region peak time of
interest, whose maximization generally has to be performed numerically.
To derive a tractable closed-form approximation of this quantity, we adopt
the uniform concentration assumption~\cite{6649346}, whereby the
concentration within a sufficiently small \ac{Rx} region is approximated by
the concentration at its center. Hence,
\begin{equation}
p_{\mathrm{Rx}}(t)
\approx
\Delta x_{\mathrm{Rx}}\,
p_{\mathrm{sd}}(x_{\mathrm{Rx}},t).
\end{equation}
Since the factor $\Delta x_{\mathrm{Rx}}$ is independent of time, it does not
affect the peak-time maximization. We therefore consider the corresponding
first-arrival concentration
\begin{equation}
p_0(x_{\mathrm{Rx}},t)
\propto
\frac{1}{\sqrt{\sigma^2(t)}}
\exp\left(
-\frac{(x_{\mathrm{Rx}}-\mu(t))^2}{2\sigma^2(t)}
\right),
\label{eq:pointwise_direct_concentration}
\end{equation}
where $\mu(t)$ and $\sigma^2(t)$ are the time-variant mean and variance
defined in~\eqref{eq:mean_var_def}.

\begin{proposition}
\label{prop:first_peak_approx}

Let
\begin{equation}
t_0
=
\frac{
\sqrt{
\left(D_{\mathrm{eff}}^{\mathrm{const}}\right)^2
+
\bar u^2 x_{\mathrm{Rx}}^2
}
-
D_{\mathrm{eff}}^{\mathrm{const}}
}
{\bar u^2}
\label{eq:t0_first_peak}
\end{equation}
denote the peak time of the corresponding steady-flow local-concentration
response, as derived in~\cite[Eq.~(5)]{vakilipoor2025cam}, where
\(D_{\mathrm{eff}}^{\mathrm{const}}\) is given in
\eqref{eq:Deff_const}. This peak time serves as the expansion point for the
pulsatile case. Furthermore, define
\begin{equation}
\kappa=\frac{R^2}{48D},
\qquad
a_0=x_{\mathrm{Rx}}-\bar u t_0,
\qquad
\sigma_0^2=2D_{\mathrm{eff}}^{\mathrm{const}}t_0,
\label{eq:peak_auxiliary_defs}
\end{equation}
and introduce the dimensionless pulsatile component
\begin{equation}
q(t)
=
\sum_{n=1}^{N}
M_n\cos(n\omega t+\phi_n),
\qquad
Q(t)
=
\int_0^t q(\tau)\,\mathrm{d}\tau,
\label{eq:q_Q_defs}
\end{equation}
so that \eqref{eq:u_t} can be written as
\(u(t)=\bar u[1+q(t)]\). Consequently, the mean displacement in
\eqref{eq:mean_var_def} can be expressed as
\begin{equation}
\mu(t)
=
\int_0^t u(\tau)\,\mathrm{d}\tau
=
\bar u[t+Q(t)].
\label{eq:mu_qQ}
\end{equation}
Finally, define
\begin{equation}
q_0=q(t_0),
\qquad
Q_0=Q(t_0).
\label{eq:q0_Q0_defs}
\end{equation}

Under the uniform concentration assumption introduced above, applied over the
\ac{Rx} length, and retaining only first-order terms in the pulsatile
amplitudes \(M_n\), the finite-region \ac{Rx} peak time \(t_{\mathrm{peak}}\) defined
in \eqref{eq:t_peak_definition} can be approximated by
\begin{equation}
t_{\mathrm{peak}}^{(1)}
=
t_0
-
\frac{F_1(t_0)}{F_0'(t_0)},
\label{eq:tpeak_first_order}
\end{equation}
with
\begin{equation}
F_0'(t_0)
=
-2D_{\mathrm{eff}}^{\mathrm{const}}
\sqrt{
\left(D_{\mathrm{eff}}^{\mathrm{const}}\right)^2
+
\bar u^2x_{\mathrm{Rx}}^2
},
\label{eq:F0prime_peak}
\end{equation}
and
\begin{align}
F_1(t_0)
&=
\bar u q_0
\left[
a_0\sigma_0^2
+
2\kappa\bar u
\left(
a_0^2-\sigma_0^2
\right)
\right]
\nonumber\\
&\quad
+
\bar u Q_0
\left[
-\bar u\sigma_0^2
+
\left(
4\kappa\bar u^2
-
2D_{\mathrm{eff}}^{\mathrm{const}}
\right)a_0
\right.
\nonumber\\
&\hspace{4.2cm}\left.
-
4\kappa\bar u D_{\mathrm{eff}}^{\mathrm{const}}
\right].
\label{eq:F1_peak}
\end{align}
\end{proposition}

\noindent\textit{Proof:} 
\vspace{0.5em}
The proof of Proposition~\ref{prop:first_peak_approx}
is provided in Appendix~\ref{app:first_peak_derivation}. Approximation~\eqref{eq:tpeak_first_order}
shows that the peak shift depends on two quantities: the instantaneous modulation \(q_0\) and
the accumulated modulation \(Q_0\). The term \(q_0\) captures the influence
of the instantaneous pulsatile velocity at the steady-flow peak time, whereas
\(Q_0\) captures the accumulated displacement induced by the pulsatile
waveform up to this time.

The approximation is most accurate when the pulsatile modulation is small to
moderate, such that terms that are quadratic or higher order in the harmonic
amplitudes \(M_n\) can be neglected. Moreover, it characterizes the first-arrival peak, while later
recirculating peaks in the closed-loop response are captured by the full
wrapped-Normal solution. The accuracy of
\eqref{eq:tpeak_first_order} is evaluated numerically in
Section~\ref{subsec:first_peak_approx_results}.

\section{Simulation Results}\label{sec:results}

In this section, we validate the proposed analytical model by comparison with
\ac{3D} \ac{PBS} and evaluate the accuracy of the proposed first-arrival
peak-time approximation. For validation, we consider the \ac{Rx} response in \eqref{eq:pRx} for both
the straight-duct model in \eqref{eq:gaussian_infinite} and the closed-loop
model in \eqref{eq:cir}. Moreover, we compare the proposed time-variant model
with the corresponding steady-flow model to assess the impact of pulsatility
on the received signal.

First, we introduce the considered flow-velocity waveforms based on
\eqref{eq:u_t} and the \ac{PBS} setup used for validation. We then analyze
the resulting received signals for different system parameters and compare
the proposed time-variant analytical model with both \ac{PBS} results and the
corresponding steady-flow analytical model from~\cite{vakilipoor2025cam}.
Finally, we evaluate the accuracy of the closed-form first-order
approximation for the first-arrival peak time.

\subsection{Velocity Waveforms}
\label{sec:Vel_waveforms}

For validation, we consider both synthetic and physiologically motivated pulsatile flow velocity waveforms in order to assess the model under controlled as well as realistic flow conditions. For all considered waveforms, the \ac{1D} velocity is given by \eqref{eq:u_t}. The corresponding \ac{3D} flow field used in the \ac{PBS} is constructed according to the Womersley model in Section~\ref{subsec:womersley_flow}, ensuring that the \ac{3D} velocity field has the same temporal mean and harmonic content.

\subsubsection{Synthetic Velocity Waveforms}
We consider two periodic synthetic waveforms. First, a sinusoidal waveform, given by
\begin{equation}
u_{\mathrm{sin}}(t)=\bar{u}\left(1+A\sin(2\pi f t)\right),
\label{eq:u_sin_sim}
\end{equation}
is studied, where $A=0.5$ denotes the oscillation amplitude and $f$ is the pulsation frequency. This corresponds to a single-harmonic representation of \eqref{eq:u_t} with $N=1$, $M_1=A$, and $\phi_1=-\pi/2$.

The second synthetic waveform is a pulsed waveform with duty cycle $d=0.2$ and temporal mean $\bar{u}$, defined as
\begin{equation}
u_{\mathrm{pulse}}(t)=
\begin{cases}
\frac{\bar{u}}{d}, & 0 \le \mathrm{mod}(t,T) < dT,\\
0, & dT \le \mathrm{mod}(t,T) < T,
\end{cases}
\label{eq:u_pulse_sim}
\end{equation}
where $T=1/f$ is the pulsation period, and $\mathrm{mod}(t,T)$ denotes the remainder of $t$ divided by $T$. Since this ideal rectangular waveform cannot be directly represented by the finite harmonic expansion in \eqref{eq:u_t}, it is approximated by a truncated Fourier series. In particular, the waveform is expressed in the form of \eqref{eq:u_t} by retaining the first $N$ harmonics of its Fourier series representation, where $N$ is chosen sufficiently large to adequately capture the shape of the rectangular pulse. The corresponding coefficients are obtained from the Fourier series of the rectangular waveform and are given by
\begin{equation}
M_n = \sqrt{A_n^2 + B_n^2}, \qquad
\phi_n = \operatorname{atan2}(-B_n, A_n),
\end{equation}
with
\begin{equation}
A_n = \frac{\sin(2\pi n d)}{\pi n d}, \qquad
B_n = \frac{1-\cos(2\pi n d)}{\pi n d},
\end{equation}
for $n=1,\dots,N$. Hence, the pulsed waveform used in the simulations is a finite-harmonic approximation of an ideal rectangular pulse.
The synthetic sinusoidal \eqref{eq:u_sin_sim} and pulsed \eqref{eq:u_pulse_sim} velocity waveforms are illustrated in Fig.~\ref{fig:waveforms_all}, for \(\bar{u}=\qty{1e-4}{\meter\per\second}\) and \(f=\qty{0.5}{\hertz}\).

\subsubsection{Realistic Velocity Waveform}
In addition to the synthetic waveforms, we consider a physiologically
motivated pulsatile waveform adopted from~\cite{11157771}. We approximate
this waveform by the harmonic representation in \eqref{eq:u_t} with
\(N=12\) harmonics. The fundamental frequency is selected by minimizing the
mean-squared fitting error, while the corresponding harmonic coefficients
are obtained by linear least-squares fitting. The resulting fit, denoted by
\(u_{\mathrm{phys}}(t)\), has fundamental frequency
\(f=\qty{1.15}{\hertz}\), amplitude coefficients
\[
\begin{aligned}
(M_1,\ldots,M_{12}) = (&0.548,\,0.684,\,0.373,\,0.489,\\
&0.352,\,0.166,\,0.253,\,0.135,\\
&0.195,\,0.134,\,0.162,\,0.190),
\end{aligned}
\]
and phase shifts
\[
\begin{aligned}
(\phi_1,\ldots,\phi_{12}) = (&-0.869,\,-1.826,\,-3.009,\,3.137,\\
&1.815,\,1.944,\,1.252,\,0.727,\\
&0.287,\,-0.504,\,-0.605,\,-1.307).
\end{aligned}
\]
The fitted waveform is then normalized with respect to its temporal mean and thus preserves its shape independently of the mean velocity $\bar{u}$, which is selected according to the considered simulation scenario. The resulting fit of the physiologically motivated waveform obtained using \eqref{eq:u_t} is shown in Fig.~\ref{fig:waveforms_all} for \(\bar{u}=\qty{1e-4}{\meter\per\second}\).

Overall, the velocity model in \eqref{eq:u_t} provides a flexible framework that can capture a wide range of velocity waveforms, including both physiologically realistic pulsatile flows in in-body environments and synthetic waveforms used in controlled experimental or microfluidic setups.

\begin{figure}[t]
  \centering
 \begin{subfigure}[t]{0.95\columnwidth}
   \centering
   \includegraphics[width=\linewidth]{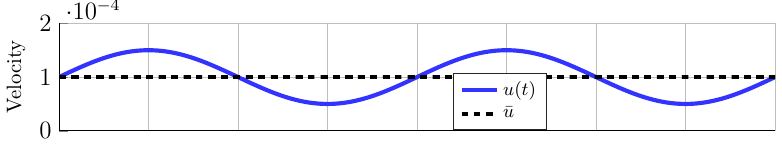}
   \caption{Sinusoidal waveform $u_{\mathrm{sin}}(t)$}
   \label{fig:waveform_sine}
 \end{subfigure}\par\vspace{0.3em}

 \begin{subfigure}[t]{0.95\columnwidth}
   \centering
   \includegraphics[width=\linewidth]{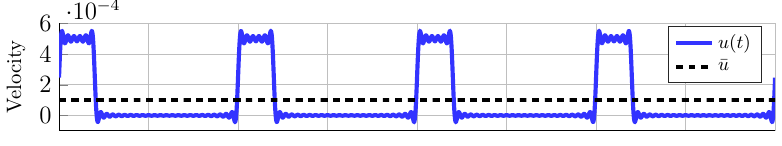}
   \caption{Pulsed waveform $u_{\mathrm{pulse}}(t)$}
   \label{fig:waveform_pulse}
 \end{subfigure}\par\vspace{0.3em}

 \begin{subfigure}[t]{0.95\columnwidth}
   \centering
  \includegraphics[width=\linewidth]{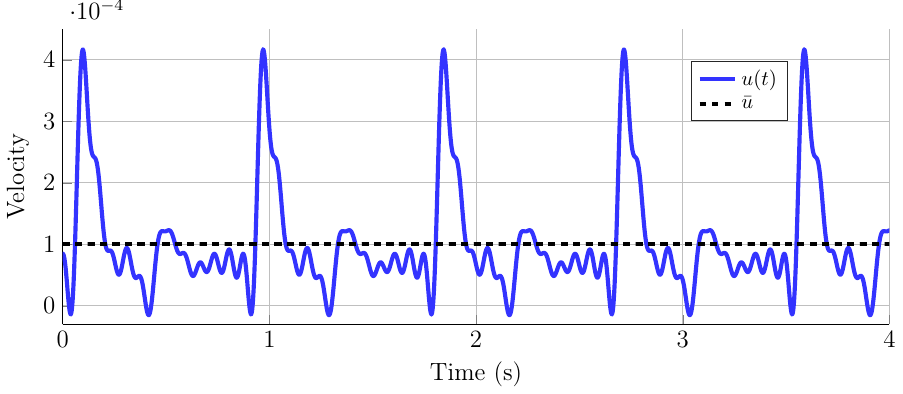}
   \caption{Physiologically motivated waveform $u_{\mathrm{phys}}(t)$}
   \label{fig:waveform_real}
 \end{subfigure}
  \caption{\textbf{Velocity waveforms} used for validating the proposed model in
\eqref{eq:gaussian_infinite}: synthetic sinusoidal (top) and pulsed (center)
waveforms, and a physiologically motivated waveform (bottom), shown for
\(\bar{u}=\qty{1e-4}{\meter\per\second}\).}
  \label{fig:waveforms_all}
\end{figure}

\subsection{Particle-Based Simulation Setup}

We implement the \ac{PBS} in a \ac{3D} cylindrical straight channel of radius
\(R\) and length \(L\). The domain is defined for \(x\in[0,L]\), and
reflecting boundary conditions are imposed at the cylindrical walls.

At time \(t=0\), \(N_{\mathrm{p}}=\num{5e5}\) particles are released
uniformly over a circular cross-section at \(x=0\).
The \ac{PBS} estimate of the finite-region \ac{Rx} response in
\eqref{eq:pRx} is obtained by counting the number of particles within a
\ac{Rx} region of length \(\Delta x_{\mathrm{Rx}}\) centered at position
\(x_{\mathrm{Rx}}\). This region corresponds to a cylindrical slice of volume
\(V_{\mathrm{Rx}}=\pi R^2\Delta x_{\mathrm{Rx}}\). The resulting estimate is
obtained as the fraction of released particles inside this \ac{Rx} region,
i.e.,
\[
p_{\mathrm{Rx}}^{\mathrm{PBS}}(t)
=
\frac{N_{\mathrm{Rx}}(t)}{N_{\mathrm p}},
\]
where \(N_{\mathrm{Rx}}(t)\) is the number of particles inside the \ac{Rx}
volume at time \(t\).

The parameters are chosen according to
\cite{Jamali2019ChannelModeling,vakilipoor2025cam} such that the system
operates in the dispersive regime and the \ac{1D} model assumptions
introduced in Section~\ref{sec:assumptions} are satisfied. Unless stated otherwise, the parameters are fixed to
\(D=\qty{5e-9}{\square\meter\per\second}\),
\(R=\qty{50}{\micro\meter}\),
\(L=\qty{1}{\milli\meter}\),
\(x_{\mathrm{Rx}}=\qty{0.3}{\milli\meter}\), and
\(\Delta x_{\mathrm{Rx}}=\qty{0.1}{\milli\meter}\).

Moreover, we assume blood-like fluid properties with density
\(\rho=\qty{1060}{\kilogram\per\meter\cubed}\) and dynamic viscosity
\(\mu=\qty{3e-3}{\pascal\second}\), corresponding to a kinematic viscosity of
\(\nu=\mu/\rho\approx\qty{2.83e-6}{\square\meter\per\second}\).
For these parameters, the resulting Womersley numbers
\(a_n\), defined in
\eqref{eq:womersley_number}, remain below unity for all considered harmonics,
indicating operation in the low-Womersley-number regime and justifying the
use of the analytical model in \eqref{eq:D1D}.

\subsection{Results for Synthetic Velocity Waveforms}

\begin{figure*}[!t]
  \centering

 \begin{subfigure}[t]{0.48\textwidth}
    \centering
    \includegraphics[width=\linewidth]{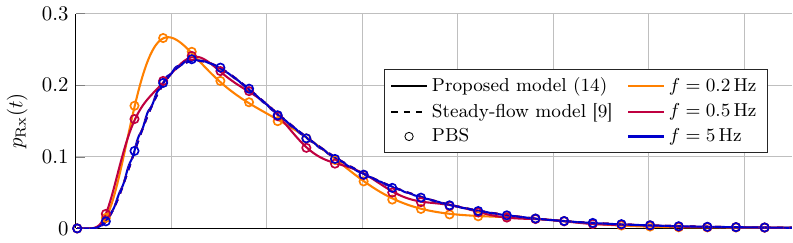}
    \caption{Sinusoidal waveform}
    \label{fig:sine_f}

    \vspace{0.4em}

    \includegraphics[width=\linewidth]{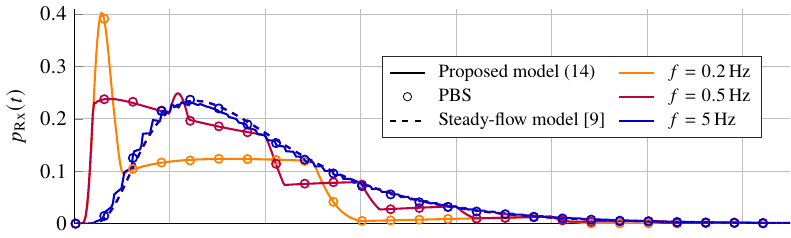}
    \caption{Pulsed waveform}
    \label{fig:pulse_f}

    \caption*{}
\end{subfigure}
\hfill
\begin{subfigure}[t]{0.48\textwidth}
    \centering

    \setcounter{subfigure}{0}

    \includegraphics[width=\linewidth]{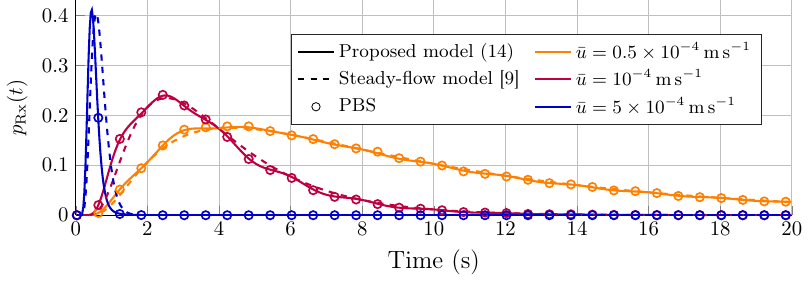}
    \caption{Sinusoidal waveform}
    \label{fig:sine_U}

    \vspace{0.4em}

    \includegraphics[width=\linewidth]{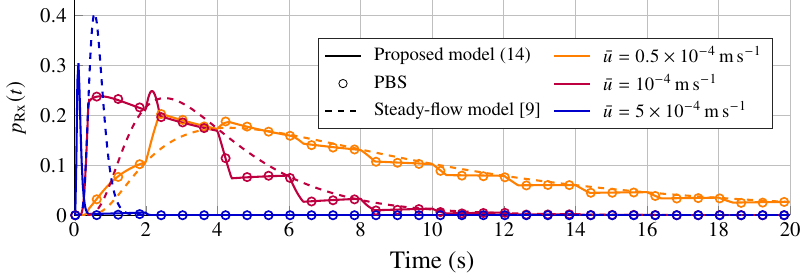}
    \caption{Pulsed waveform}
    \label{fig:pulse_U}

    \caption*{}
\end{subfigure}

  \vspace{-0.5em}

  \begin{minipage}[t]{0.48\textwidth}
    \caption{\small \textbf{Received signal for different pulsation frequencies \(f\)}, for
    sinusoidal (top) and pulsed (bottom) velocity waveforms.}
    \label{fig:synthetic_f}
  \end{minipage}
  \hfill
  \begin{minipage}[t]{0.48\textwidth}
    \caption{\small\textbf{Received signal for different mean velocities \(\bar{u}\)}, for
    sinusoidal (top) and pulsed (bottom) velocity waveforms.}
    \label{fig:synthetic_U}
  \end{minipage}

\end{figure*}

First, we investigate the synthetic velocity waveforms introduced in
Section~\ref{sec:Vel_waveforms} (cf. Fig.~\ref{fig:waveform_sine} and
Fig.~\ref{fig:waveform_pulse}), where either the pulsation frequency \(f\) or
the mean velocity \(\bar{u}\) is varied. Since our focus is on the
influence of pulsation rather than recirculation, we consider a straight-duct
scenario and evaluate the \ac{Rx} response \(p_{\mathrm{Rx}}(t)\) in
\eqref{eq:pRx} using the straight-duct solution in
\eqref{eq:gaussian_infinite}.

Figure~\ref{fig:synthetic_f} shows the received signal for different
pulsation frequencies \(f\) for both synthetic velocity waveforms, for
\(\bar{u}=\qty{1e-4}{\meter\per\second}\) (top: sinusoidal, bottom:
rectangular). Since the steady-flow model in~\cite{vakilipoor2025cam}
assumes a constant flow velocity, it does not depend on the pulsation
frequency. Therefore, it is
identical for all values of \(f\) (dashed curve). First, we observe that the results obtained
from the proposed model (solid curves) are in excellent agreement with the
\ac{PBS} results (circle markers) for all considered frequencies and both waveforms. Furthermore, it can be observed
that, for increasing \(f\), the received signal obtained from \eqref{eq:gaussian_infinite} converges to
that of the corresponding steady-flow system with constant velocity
\(\bar{u}\), confirming the analysis in Section~\ref{subsec:high_frequency_limit}.
This behavior can be
explained by a temporal averaging effect: at higher frequencies, the
pulsation period becomes small compared to the particle propagation time,
such that particles experience multiple velocity oscillations during their
propagation and are effectively affected by the mean flow velocity.

Figure~\ref{fig:synthetic_U} shows the received signal for different values
of the mean velocity \(\bar{u}\) for both synthetic waveforms, for fixed
\(f=\qty{0.5}{\hertz}\). Similar to the results in
Fig.~\ref{fig:synthetic_f}, the results obtained from the proposed model (solid curves) show excellent agreement with the results obtained from \ac{PBS} (circle markers) for all considered \(\bar{u}\) values. Moreover, as the temporal mean
velocity \(\bar{u}\) increases, the received signal peaks occur earlier and
become sharper, which is consistent with faster advective transport. In addition, more pronounced deviations from the corresponding steady-flow
model can be observed, indicating a stronger influence of flow pulsatility. This behavior can be attributed to the fact that, for larger
\(\bar{u}\), particle transport becomes increasingly dominated by advection,
such that particles more closely follow the instantaneous velocity variations
and are therefore more strongly influenced by the pulsations. In contrast,
for decreasing \(\bar{u}\), the received signal becomes increasingly similar
to that of the corresponding steady-flow system. This is because diffusion
becomes again more dominant, effectively smoothing out the
temporal variations of the flow and reducing the impact of pulsations on the
particle trajectories.
Furthermore, the comparison with the corresponding steady-flow response
shows that pulsatility can shift the peaks of the received signal. Such peak
shifts can strongly affect the choice of optimal sampling times at the \ac{Rx}
and further motivate the first-arrival peak-time approximation derived in
Section~\ref{subsec:first_peak_analytical_approx}.

\subsection{Results for Realistic Velocity Waveform}

Next, we consider the physiologically motivated velocity waveform introduced
in Section~\ref{sec:Vel_waveforms} (cf. Fig.~\ref{fig:waveform_real}), and
vary the mean velocity \(\bar{u}\), the \ac{Rx} position \(x_{\mathrm{Rx}}\),
and the diffusion coefficient \(D\). As in the previous experiments, the investigated quantity is the \ac{Rx} response \(p_{\mathrm{Rx}}(t)\) in \eqref{eq:pRx}.

Figure~\ref{fig:pbs} shows the received signal (solid curves) obtained by inserting \eqref{eq:gaussian_infinite}
in \eqref{eq:pRx}, while the circle markers show the corresponding
\ac{PBS} estimate \(p_{\mathrm{Rx}}^{\mathrm{PBS}}(t)\). The dashed curves
show the steady-flow model in~\cite{vakilipoor2025cam}. Results are shown
for different \ac{Rx} positions \(x_{\mathrm{Rx}}\) with fixed mean flow
velocity \(\bar{u}=\qty{2e-4}{\meter\per\second}\) (top plot), for
different temporal mean velocities \(\bar{u}\) with fixed \ac{Rx} position
\(x_{\mathrm{Rx}}=\qty{0.3}{\milli\meter}\) (center plot), and for different
diffusion coefficients \(D\) with fixed \ac{Rx} position
\(x_{\mathrm{Rx}}=\qty{0.3}{\milli\meter}\) and fixed temporal mean velocity
\(\bar{u}=\qty{1e-4}{\meter\per\second}\) (bottom plot).

First, we observe that the proposed model is in excellent agreement with the
\ac{PBS} results for all considered scenarios. Moreover, the steady-flow
model in~\cite{vakilipoor2025cam} approximates the overall response shape,
but it cannot capture the time variations induced by the pulsatile flow. The
influence of pulsatility becomes less pronounced for larger \ac{Rx}
distances, smaller temporal mean velocities, and larger diffusion
coefficients. These observations are consistent with those obtained for the
synthetic velocity waveforms and can be explained by the reduced relative
importance of advective transport compared to diffusion, which leads to
more dispersive dynamics~\cite{Jamali2019ChannelModeling}.

Furthermore, Fig.~\ref{fig:pbs_ubar} shows that pulsatility also leads to
noticeable peak shifts for the physiologically motivated waveform, highlighting
the importance of the proposed time-variant model for accurately characterizing
the received signal in \ac{MC} systems with pulsatile flow.

Finally, to illustrate the closed-loop behavior, Fig.~\ref{fig:closed_loop}
shows the \ac{Rx} response obtained by substituting the closed-loop solution
in \eqref{eq:cir} into \eqref{eq:pRx}. For this closed-loop result only, the
received signal is normalized by its equilibrium value
\[
p_{\mathrm{Rx},\infty}
=
\frac{\Delta x_{\mathrm{Rx}}}{L},
\]
which follows from the uniform equilibrium concentration
\(\lim_{t\to\infty}p(x,t)=1/L\). Hence, the normalized closed-loop response
\(p_{\mathrm{Rx}}(t)/p_{\mathrm{Rx},\infty}\) approaches unity as
\(t\to\infty\). In contrast to the straight-duct \ac{Rx} response considered above, the
closed-loop response exhibits additional delayed peaks corresponding to
molecules recirculating around the loop and returning to the \ac{Rx} region.
Moreover, the influence of pulsatility becomes less pronounced for successive
circulations. This behavior can be attributed to the increasing temporal
spreading of the molecule distribution with propagation time, which smooths
the variations induced by the pulsatile flow velocity. The close agreement
between the analytical closed-loop response and the corresponding \ac{PBS}
results demonstrates that the proposed model can accurately characterize both
the initial and recirculating molecule arrivals with pulsatile flow.

\begin{figure}[!t]
  \centering
 \begin{subfigure}[t]{0.95\columnwidth}
   \centering
   \includegraphics[width=\linewidth]{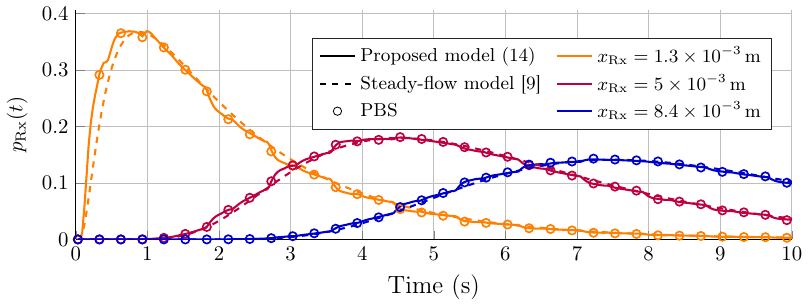}
   \caption{Different \(x_{\mathrm{Rx}}\) for fixed
   \(\bar{u}=\qty{2e-4}{\meter\per\second}\).}
   \label{fig:pbs_xdet}
 \end{subfigure}\par\vspace{0.3em}
 \begin{subfigure}[t]{0.95\columnwidth}
   \centering
   \includegraphics[width=\linewidth]{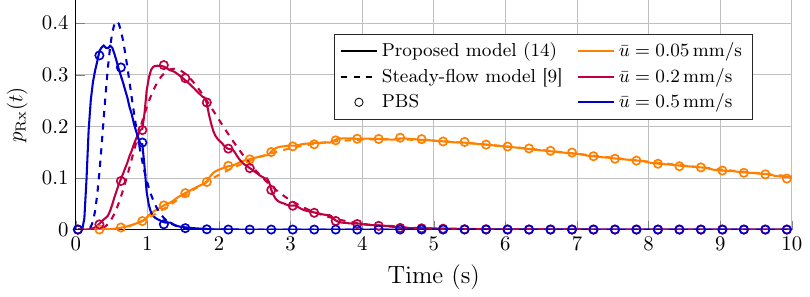}
   \caption{Different \(\bar{u}\) for fixed
   \(x_{\mathrm{Rx}}=\qty{0.3}{\milli\meter}\).}
   \label{fig:pbs_ubar}
 \end{subfigure}\par\vspace{0.3em}
 \begin{subfigure}[t]{0.95\columnwidth}
   \centering
   \includegraphics[width=\linewidth]{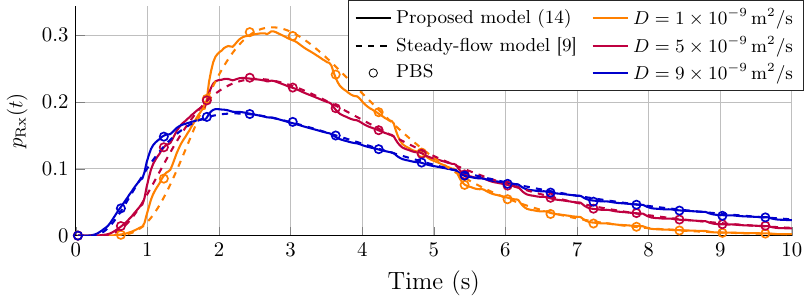}
   \caption{Different \(D\) for fixed
   \(x_{\mathrm{Rx}}=\qty{0.3}{\milli\meter}\) and
   \(\bar{u}=\qty{1e-4}{\meter\per\second}\).}
   \label{fig:pbs_D}
 \end{subfigure}
  \caption{\textbf{Physiologically motivated pulsatile flow:} Received
  signals for different \ac{Rx} positions \(x_{\mathrm{Rx}}\), temporal mean
  velocities \(\bar{u}\), and diffusion coefficients \(D\).}
  \label{fig:pbs}
\end{figure}

\begin{figure}[t]
    \centering
    \includegraphics[width=\columnwidth]{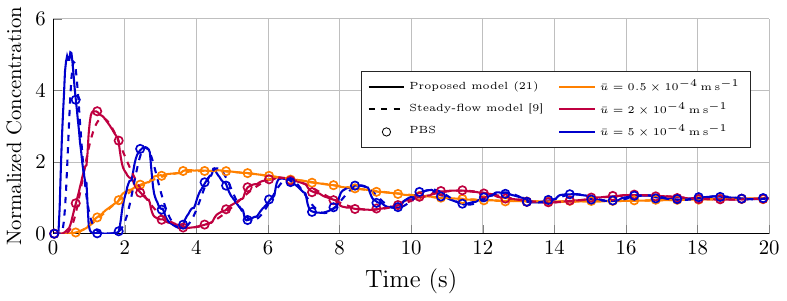}
    \caption{\textbf{Closed-loop physiologically motivated pulsatile flow:}
    Received signals for different temporal mean velocities \(\bar{u}\), comparing
    the proposed analytical model with \ac{PBS}.}
    \label{fig:closed_loop}
\end{figure}

\subsection{Accuracy of the First-Arrival Peak-Time Approximation}
\label{subsec:first_peak_approx_results}

We evaluate the accuracy of the first-order peak-time approximation in
\eqref{eq:tpeak_first_order} by comparing it with the finite-region \ac{Rx}
first-arrival peak time \(t_{\mathrm{peak}}\) defined in
\eqref{eq:t_peak_definition}. For each considered parameter setting,
\(t_{\mathrm{peak}}\) is obtained numerically by maximizing the finite-region
\ac{Rx} response obtained by substituting the straight-duct solution in
\eqref{eq:gaussian_infinite} into \eqref{eq:pRx}, whereas
\(t_{\mathrm{peak}}^{(1)}\) is obtained directly from the closed-form
approximation in \eqref{eq:tpeak_first_order}.

To investigate the agreement under representative parameter variations,
we consider three controlled sweeps within the applicability regime of
the proposed model characterized in
Section~\ref{sec:regime_analysis}. In each sweep, one model parameter is
varied while the remaining parameters are kept fixed.
Figure~\ref{fig:first_peak_approximation} shows the resulting peak times as
functions of (a) the mean flow velocity \(\bar{u}\), (b) the diffusion
coefficient \(D\), and (c) the receiver distance \(x_{\mathrm{Rx}}\).
In each panel, the solid curve represents the reference finite-region \ac{Rx} peak
time \(t_{\mathrm{peak}}\), whereas the dashed curve represents the
corresponding first-order approximation \(t_{\mathrm{peak}}^{(1)}\).

For each controlled sweep, the overall discrepancy between the two peak
times is summarized by the root-mean-square error
\begin{equation}
\mathrm{RMSE}
=
\sqrt{
\frac{1}{K}
\sum_{i=1}^{K}
\left(
t_{\mathrm{peak},i}^{(1)}
-
t_{\mathrm{peak},i}
\right)^2
},
\label{eq:first_peak_rmse}
\end{equation}
where \(K\) denotes the number of values included in the
corresponding sweep.

As shown in Fig.~\ref{fig:first_peak_approximation}, the first-order
approximation closely follows the reference peak time over the considered
parameter variations. In Fig.~\ref{fig:first_peak_approximation}a,
increasing the mean velocity shifts the first-arrival peak to earlier
times, consistent with the shorter advective propagation time, while the
approximate and reference peak times remain close throughout the sweep.
Figure~\ref{fig:first_peak_approximation}b shows the dependence on the
diffusion coefficient. The peak time exhibits a weak dependence on
\(D\) over the considered range, and the first-order approximation closely
reproduces this behavior. Finally,
Fig.~\ref{fig:first_peak_approximation}c shows the effect of increasing
the receiver distance. The peak time increases with
\(x_{\mathrm{Rx}}\), reflecting the longer propagation distance, while the
two curves remain nearly overlapping over the considered range.

The small differences between the two peak times, together with the RMSE values
reported in the individual panels, demonstrate that the
first-order approximation \eqref{eq:tpeak_first_order} provides an
accurate estimate of the actual first-arrival peak time
\(t_{\mathrm{peak}}\). Moreover, the RMSE values are of the same order of
magnitude for all three parameter sweeps, ranging from
\(0.064\,\mathrm{s}\) to \(0.095\,\mathrm{s}\).
Thus, the approximation provides a computationally efficient alternative to
explicitly searching for the first-arrival maximum of the finite-region
\ac{Rx} response in \eqref{eq:pRx}.

\begin{figure*}[!t]
    \centering
    \includegraphics[width=\textwidth]{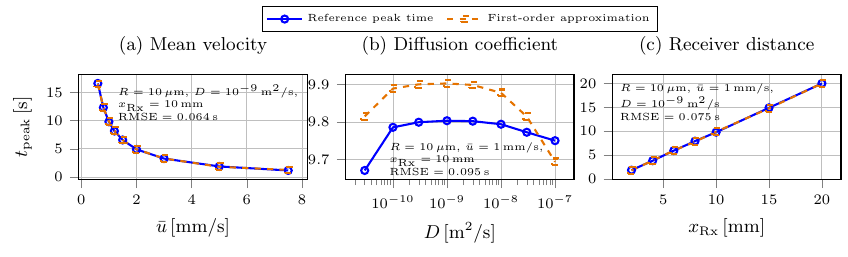}
    \caption{\textbf{First-arrival peak-time approximation.} Comparison between the reference finite-region first-arrival peak time \(t_{\mathrm{peak}}\) and its first-order approximation \(t_{\mathrm{peak}}^{(1)}\) for varying (a) mean velocity \(\bar{u}\), (b) diffusion coefficient \(D\), and (c) receiver distance \(x_{\mathrm{Rx}}\). The RMSE is reported in each panel.}
    \label{fig:first_peak_approximation}
\end{figure*}

\begin{figure}[t]
    \centering
    \includegraphics[width=0.95\columnwidth]{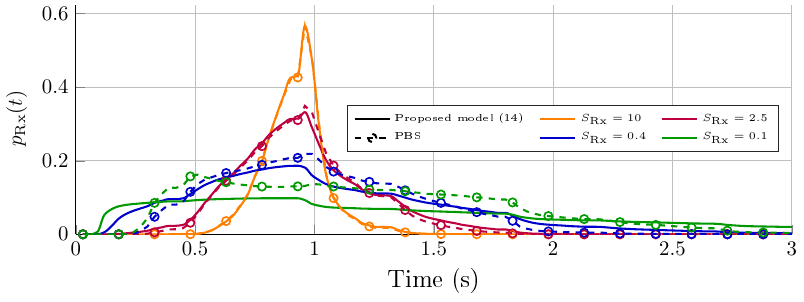}
    \caption{\textbf{PBS-based validation of the admissibility criterion.}
Comparison of the analytical and \ac{PBS} received signals for different
values of \(S_{\mathrm{Rx}}\).}
    \label{fig:straight_first_peak_radius}
\end{figure}

\section{Regime Analysis}
\label{sec:regime_analysis}

In this section, we assess the applicability of the proposed time-variant
analytical model by investigating when the reduced \ac{1D} model in
\eqref{eq:adv_diff} provides an accurate approximation of the \ac{3D}
transport model in \eqref{eq:3D_adv_diff}. First, we use \ac{3D} \ac{PBS}
over a broad parameter range to quantify the accuracy of the reduced model
and introduce the nondimensional admissibility indicator
\(S_{\mathrm{Rx}}\), which summarizes how strongly the reduced-model
assumptions are satisfied for a given \ac{Tx}-\ac{Rx} distance. Based on the
parameter sweep, we identify a practical admissibility criterion and use it
to assess model applicability beyond the parameter combinations directly
validated by \ac{PBS}. Finally, we discuss the applicability of the model to
representative blood-vessel classes, whose compatibility with the physical
model assumptions is summarized qualitatively in a table.

\subsection{PBS-Based Model Applicability Analysis}
\label{subsec:numerical_regime_validation}

To support the applicability analysis, we performed an extensive
\ac{PBS}-based parameter sweep over vessel radii, mean flow velocities,
diffusion coefficients, and \ac{Tx}-\ac{Rx} distances. The proposed
analytical model was compared with the corresponding \ac{3D} \ac{PBS}
response, and the discrepancy was quantified by the RMSE. Since the results of the parameter-sweep analysis for all considered
parameter settings cannot be shown in the paper, additional received-signal
figures from the parameter sweep are provided in an accompanying
repository~\cite{Symeonidis2026RegimeAnalysis}. In the following, we use these results to motivate a compact admissibility indicator that summarizes how strongly the reduced-model assumptions are satisfied over the \ac{Tx}-\ac{Rx} distance, and we show
representative cases that illustrate its interpretation.

To assess whether the reduced \ac{1D} model is applicable at the \ac{Rx}, we
evaluate the assumptions in Section~\ref{sec:assumptions} given the propagation
distance from the \ac{Tx} to the \ac{Rx}. Accordingly, the characteristic
axial length scale entering these conditions is taken as
\(x_{\mathrm{Rx}}\). The conditions in
\eqref{eq:slender}--\eqref{eq:axial_dispersion_condition} can then be
rewritten as
\begin{equation}
\frac{x_{\mathrm{Rx}}}{R} \gg 1,\qquad
\frac{D x_{\mathrm{Rx}}}{\bar{u}R^2} \gg 1,\qquad
\frac{\bar{u}x_{\mathrm{Rx}}}{D} \gg 1.
\label{eq:nondimensional_conditions}
\end{equation}

These measures correspond, respectively, to the slenderness of the
propagation domain, the separation between radial diffusion and axial
advection time scales, and the dominance of axial advection over axial
molecular diffusion.

Since the proposed \ac{1D} model requires all three conditions to be satisfied, no single measure is sufficient to characterize the
applicability regime. We therefore define the conservative nondimensional
admissibility indicator
\begin{equation}
S_{\mathrm{Rx}}
=
\min\left(
\frac{x_{\mathrm{Rx}}}{R},
\frac{D x_{\mathrm{Rx}}}{\bar u R^2},
\frac{\bar u x_{\mathrm{Rx}}}{D}
\right).
\label{eq:S_Rx_regime}
\end{equation}
The minimum is used because the least satisfied condition limits the validity
of the reduced \ac{1D} approximation. Thus, larger values of
\(S_{\mathrm{Rx}}\) indicate that the combined reduced-model requirements are
more strongly satisfied for a given \ac{Tx}-\ac{Rx} arrangement.

The accompanying dataset~\cite{Symeonidis2026RegimeAnalysis} was used to identify the operating
regime in which the reduced \ac{1D} model remains accurate. The resulting RMSE values were then related to the corresponding values of
\(S_{\mathrm{Rx}}\). This allowed us to quantify how
strongly the reduced-model assumptions have to be satisfied in practice. The
sweep results support \(S_{\mathrm{Rx}}\ge2\) as a practical admissibility
criterion for parameter combinations for which the proposed analytical model
is expected to accurately reproduce the dynamics of the received signal. 

Figure~\ref{fig:straight_first_peak_radius} compares the proposed analytical
model (solid curves) with \ac{PBS} (dashed curves with markers) for different values of \(S_{\mathrm{Rx}}\). For this
comparison, the \ac{Rx} distance, mean velocity, and diffusion coefficient are
fixed to \(x_{\mathrm{Rx}}=\qty{0.50}{\milli\meter}\),
\(\bar u=\qty{0.5}{\milli\meter\per\second}\), and
\(D=\qty{1e-9}{\square\meter\per\second}\), respectively, while the vessel
radius is varied, i.e., \(R=\qty{10}{\micro\meter}\),
\(\qty{20}{\micro\meter}\), \(\qty{50}{\micro\meter}\), and
\(\qty{100}{\micro\meter}\). The corresponding values of the admissibility
indicator in \eqref{eq:S_Rx_regime} are \(S_{\mathrm{Rx}}=10\), \(2.5\),
\(0.4\), and \(0.1\), respectively. The cases \(R=\qty{10}{\micro\meter}\) and
\(R=\qty{20}{\micro\meter}\) satisfy the criterion
\(S_{\mathrm{Rx}}\ge2\), and the proposed analytical model is in excellent
agreement with the \ac{PBS} results. As the vessel radius increases further, however,
the admissibility criterion \(S_{\mathrm{Rx}}\ge2\) is no longer satisfied and the agreement
between the analytical and \ac{PBS} responses deteriorates. This behavior is mainly caused by the second nondimensional measure in
\eqref{eq:nondimensional_conditions}, which follows from the radial-diffusion
condition in \eqref{eq:radial_fast} and decreases with \(1/R^2\).

Overall, the validation supports the use of \(S_{\mathrm{Rx}}\ge2\) as a
practical admissibility criterion for identifying parameter combinations for which
the reduced \ac{1D} model in \eqref{eq:adv_diff} is expected to provide an
accurate approximation of the \ac{3D} model in \eqref{eq:3D_adv_diff}.
The \ac{PBS} validation above was performed for comparably short
\ac{Tx}-\ac{Rx} distances, for which direct comparison with \ac{3D} \ac{PBS}
is computationally feasible. At larger propagation distances, \ac{PBS}
becomes increasingly demanding because particles have to be tracked over
longer times. In contrast, \(S_{\mathrm{Rx}}\) can be evaluated directly for
such scenarios. Since all three nondimensional measures in
\eqref{eq:nondimensional_conditions} increase with \(x_{\mathrm{Rx}}\) for
fixed vessel radius \(R\), mean velocity \(\bar u\), and diffusion coefficient
\(D\), the reduced-model assumptions become more strongly satisfied as the
\ac{Tx}-\ac{Rx} distance increases.

\renewcommand{\arraystretch}{1.2}
\begin{table*}
  \centering
 \caption{Representative blood-vessel classes and qualitative compatibility with the model assumptions.}\label{tab:Vessel_Types}
  \resizebox{\textwidth}{!}{%
  \begin{tabular}{|l||c|c|c|c|c|c||c|}
    \hline
    \rowcolor{gray!40}\textbf{Blood vessel} & \textbf{Radius} $R$ & \textbf{Time-averaged velocity $\bar{u}$} & \textbf{Rigid walls?} & \textbf{Tapering?} & \textbf{Newtonian fluid?} & \textbf{Flow pulsatility} & \textbf{Reference} \\ \hline\hline

    \multicolumn{8}{|c|}{\cellcolor{gray!20}\textbf{Large conduit arteries}} \\ \hline\hline
    Aorta & \(\sim\qtyrange{10.5}{18.5}{\milli\meter}\) & \(\sim\qtyrange{30}{100}{\centi\meter\per\second}\) & \cellcolor{red!10} No & \cellcolor{red!10} Yes & \cellcolor{green!10} Yes & Strong & \cite{Pedley1980, Aaronson2012, Wolak2008, Sidebotham2007} \\ \hline
    Carotid, iliac & $\sim[2.8,3.5]\,\SI{}{\milli\meter}$ & $\sim[30,45]\,\SI{}{\centi\meter\per\second}$ & \cellcolor{red!10} No & \cellcolor{yellow!20} Marginal & \cellcolor{green!10} Yes & Strong & \cite{Ferrara1995, Pomella2017}  \\ \hline
    Femoral, renal & $\sim[2.4,4.9]\,\SI{}{\milli\meter}$ & $\sim[12,81]\,\SI{}{\centi\meter\per\second}$ & \cellcolor{yellow!20} Marginal & \cellcolor{yellow!20} Marginal & \cellcolor{green!10} Yes & Strong & \cite{Greene1981, Sandgren1999, Lewis1986}  \\ \hline\hline

    \multicolumn{8}{|c|}{\cellcolor{gray!20}\textbf{Medium-sized arteries}} \\ \hline\hline
    Small named arteries & $\sim[0.5,3]\,\SI{}{\milli\meter}$ & $\sim[10,45]\,\SI{}{\centi\meter\per\second}$ & \cellcolor{yellow!20} Marginal & \cellcolor{yellow!20} Marginal & \cellcolor{green!10} Yes & Moderate & \cite{Aaronson2012}  \\ \hline\hline

    \multicolumn{8}{|c|}{\cellcolor{gray!20}\textbf{Arterioles}} \\ \hline\hline
    Large arterioles & \(\sim\qtyrange{15}{50}{\micro\meter}\) & $\sim[1,10]\,\SI{}{\centi\meter\per\second}$ & \cellcolor{green!10} Yes & \cellcolor{green!10} No & \cellcolor{yellow!20} Marginal & Moderate & \cite{Pedley1980, Aaronson2012}  \\ \hline
    Terminal arterioles & $\sim[3,15]\,\SI{}{\micro\meter}$ & \(\sim\qtyrange{1.2}{5}{\milli\meter\per\second}\) & \cellcolor{green!10} Yes & \cellcolor{green!10} No & \cellcolor{red!10} No & Weak & \cite{Pedley1980, Aaronson2012, Koutsiaris2016, Dremin2025}  \\ \hline \hline

    \multicolumn{8}{|c|}{\cellcolor{gray!20}\textbf{Capillaries}} \\ \hline\hline
    Capillaries & $\sim[2,2.5]\,\SI{}{\micro\meter}$ & $\sim[0.6,1.2]\,\SI{}{\milli\meter\per\second}$ & \cellcolor{green!10} Yes & \cellcolor{green!10} No & \cellcolor{red!10} No & Weak & \cite{Pedley1980, Aaronson2012, Dremin2025, Koutsiaris2016}  \\ \hline \hline

    \multicolumn{8}{|c|}{\cellcolor{gray!20}\textbf{Venules and veins}} \\ \hline\hline
    Venules & $\sim[3.5,200]\,\SI{}{\micro\meter}$ & $\sim[0.1,2.5]\,\SI{}{\centi\meter\per\second}$ & \cellcolor{green!10} Yes & \cellcolor{green!10} No & \cellcolor{yellow!20} Marginal & Weak & \cite{Aaronson2012, Gavins2016, Ricci2011} \\ \hline
    Small veins & $\sim[0.2,2]\,\SI{}{\milli\meter}$ & $\sim[2.5,7.5]\,\SI{}{\centi\meter\per\second}$ & \cellcolor{green!10} Yes & \cellcolor{green!10} No & \cellcolor{green!10} Yes & Weak & \cite{Aaronson2012}  \\ \hline
    Medium-sized veins & $\sim[2,5]\,\SI{}{\milli\meter}$ & $\sim[4,17]\,\SI{}{\centi\meter\per\second}$ & \cellcolor{yellow!20} Marginal & \cellcolor{green!10} No & \cellcolor{green!10} Yes & Moderate & \cite{Aaronson2012, Valdueza1996}  \\ \hline
    Large veins, vena cava & $\sim[6.5,15]\,\SI{}{\milli\meter}$ & $\sim[15,40]\,\SI{}{\centi\meter\per\second}$ & \cellcolor{yellow!20} Marginal & \cellcolor{yellow!20} Marginal & \cellcolor{green!10} Yes & Strong & \cite{Aaronson2012, Prince1983}  \\ \hline \hline

    \multicolumn{8}{|c|}{\cellcolor{gray!20}\textbf{Pulmonary vessels}} \\ \hline\hline
    Pulmonary arteries & $\sim[13.25,23.3]\,\SI{}{\milli\meter}$ & $\sim[10,100]\,\SI{}{\centi\meter\per\second}$ & \cellcolor{red!10} No & \cellcolor{red!10} Yes & \cellcolor{green!10} Yes & Strong & \cite{Berger2022, Gardin1988}  \\ \hline
    Pulmonary arterioles & $\sim[10,50]\,\SI{}{\micro\meter}$ & \(\sim\qtyrange{1.2}{5}{\milli\meter\per\second}\) & \cellcolor{green!10} Yes & \cellcolor{green!10} No & \cellcolor{yellow!20} Marginal & Moderate & \cite{Townsley2012, Koutsiaris2016} \\ \hline
    Pulmonary venules & $\sim[10,50]\,\SI{}{\micro\meter}$ & \(\sim\qtyrange{1.2}{5}{\milli\meter\per\second}\) & \cellcolor{green!10} Yes & \cellcolor{green!10} No & \cellcolor{yellow!20} Marginal & Weak & \cite{Townsley2012, Koutsiaris2016} \\ \hline
    Pulmonary capillaries & $\sim[3.5,5]\,\SI{}{\micro\meter}$ & $\sim[0.5,1]\,\SI{}{\milli\meter\per\second}$ & \cellcolor{green!10} Yes & \cellcolor{green!10} No & \cellcolor{red!10} No & Weak & \cite{Chambers2019} \\ \hline
  \end{tabular}
  }\vspace*{1mm}

  {\setlength{\fboxsep}{0pt}%
  \fcolorbox{black}{green!10}{\makebox[2ex][c]{\rule{0pt}{2ex}}}} Supports model assumptions\hspace*{1cm}
  {\setlength{\fboxsep}{0pt}%
  \fcolorbox{black}{yellow!20}{\makebox[2ex][c]{\rule{0pt}{2ex}}}} Marginal\hspace*{1cm}
  {\setlength{\fboxsep}{0pt}%
  \fcolorbox{black}{red!10}{\makebox[2ex][c]{\rule{0pt}{2ex}}}} Violates model assumptions
\end{table*}
\renewcommand{\arraystretch}{1.0}

\subsection{Applicability to Representative Blood-Vessel Classes}
\label{subsec:vessel_class_applicability}

The indicator \(S_{\mathrm{Rx}}\) summarizes the nondimensional transport
conditions associated with the reduced \ac{1D} model. However, physiological
blood vessels may also deviate from the idealized physical assumptions used
in the derivation. In particular, the proposed model assumes a local straight
cylindrical vessel segment with approximately constant radius, rigid walls,
Newtonian fluid behavior, and laminar fully developed flow. Therefore, in
addition to the admissibility criterion, the physical properties of
different vessel classes have to be considered. This assessment concerns the local vessel-segment transport model and,
in particular, the direct propagation from the \ac{Tx} to the \ac{Rx}
before recirculation. Accordingly, Table~\ref{tab:Vessel_Types} evaluates
the compatibility of individual vessel classes with the local transport
assumptions, rather than the topology of the vascular network as a whole.
This distinction is important because vessel size and flow pulsatility may
vary substantially between individual links of the network.

Table~\ref{tab:Vessel_Types} provides a qualitative assessment of
representative blood-vessel classes with respect to these assumptions. The
table includes vessel size, typical time-averaged velocity, wall rigidity,
tapering, Newtonian fluid behavior, and the expected flow pulsatility. The tapering column is included because the analytical model assumes an
approximately constant vessel radius over the considered local segment.
Therefore, pronounced tapering represents a deviation from the assumed
cylindrical geometry. Mild tapering, however, does not necessarily preclude
the use of the model, since a sufficiently short vessel segment may still be
approximated as locally cylindrical with an approximately constant radius.
Therefore, tapering is treated as a secondary applicability consideration
rather than as a standalone exclusion criterion. Other deviations from the
model assumptions, such as pronounced vessel-wall deformation in response to pressure variations, non-Newtonian blood
rheology, and very small vessel dimensions, may impose more fundamental
limitations on the proposed transport model.

The flow pulsatility classifications in Table~\ref{tab:Vessel_Types} are
qualitative and indicate the expected importance of temporal flow variations
for particle transport. Strong pulsatility denotes cases in which the
time variation of the flow should be explicitly accounted for, whereas
moderate pulsatility indicates that a steady-flow approximation may be
reasonable but can introduce some loss of accuracy. Weak pulsatility denotes
cases in which temporal flow variations are sufficiently small, such that particle
transport can generally be modeled using a steady-flow approximation.
These classifications describe the importance of pulsatile effects and
should not be interpreted as a criterion for the applicability of the
reduced \ac{1D} transport model itself.

The classifications in Table~\ref{tab:Vessel_Types} are qualitative and are
based on the physiological characteristics reported in the cited literature.
The considered physical properties, such as wall rigidity, tapering, and
Newtonian fluid behavior, vary continuously rather than forming strictly
binary categories. Accordingly, ``marginal'' denotes the transitional range
between clear support for and clear violation of a model assumption. In such
cases, the corresponding property is only partially compatible with the
idealized assumption and may depend on the specific vessel segment or
physiological conditions. The marginal classification is therefore not a
quantitative threshold, but rather a qualitative guide for identifying cases
in which the corresponding model assumption may be conditionally reasonable.

Several vessel classes in Table~\ref{tab:Vessel_Types} therefore provide
promising application regimes for the proposed analytical model. Large
arterioles and pulmonary arterioles are particularly relevant examples. They
support the assumptions of approximately rigid walls and negligible tapering,
while their flow pulsatility is classified as moderate. However, the Newtonian-fluid approximation is marked as marginal in
these vessel classes, which suggests that the proposed model should be
interpreted as an effective transport description.
Nevertheless, when the admissibility criterion
\(S_{\mathrm{Rx}}\ge2\) is satisfied, these vessel classes are suitable
candidates for modeling residual pulsatile-flow effects over local
\ac{Tx}-\ac{Rx} distances in small vessels.

Small veins, venules, and pulmonary venules also satisfy several of the main
physical assumptions of the model. In particular, they are compatible with
the assumptions of approximately rigid walls and negligible tapering, and
small veins additionally support the Newtonian-fluid assumption. Their flow pulsatility is classified as weak, but this does not invalidate the model.
Instead, it implies that the time variation of the flow has a weaker impact
on the \ac{Rx} response. These vessel classes are
therefore compatible with the proposed model, although pronounced
pulsatile-flow effects are not expected.

Medium-sized veins and large veins, including the vena cava, also form
compatible vessel classes. They support the Newtonian-fluid assumption and
exhibit moderate or strong pulsatility. Their wall rigidity, and in large
veins also their local geometric uniformity, are marked as marginal. These
marginal properties do not necessarily preclude the use of the proposed model
as a local vessel-segment approximation over sufficiently uniform segments,
provided that the admissibility criterion \(S_{\mathrm{Rx}}\ge2\) is
satisfied. In such cases, the time-variant formulation may be relevant,
particularly when residual pulsatility is not negligible.

Arteries below the largest conduit arteries, including femoral, renal, and
medium-sized arteries, are also conditionally compatible with the proposed
model. These vessels generally support the Newtonian-fluid assumption and
retain moderate to strong pulsatility. Although wall rigidity and local
geometric uniformity are not always fully satisfied, they are marked as
marginal rather than as clear violations. Hence, the proposed model can be
used as a local vessel-segment approximation in arterial segments where wall
deformation and spatial variations of the vessel radius are sufficiently
small, again provided that the admissibility criterion \(S_{\mathrm{Rx}}\ge2\)
is satisfied.

The least compatible vessel classes are the largest conduit arteries,
including the aorta and pulmonary arteries. These vessels are highly relevant
from the viewpoint of pulsatile flow and support the Newtonian-fluid
assumption, but they violate the rigid-wall assumption and exhibit more
complex hemodynamics. Therefore, the proposed model should not be interpreted
as a complete transport model for these vessels. Instead, it can only be used as a local vessel-segment approximation, while
a more complete model would need to account for arterial wall compliance and
its effects on pulsatile hemodynamics, including Windkessel-type buffering
and pulse-wave propagation.

Finally, terminal arterioles, capillaries, and pulmonary capillaries are
mainly limited by the Newtonian-fluid assumption. Their rigid-wall and
constant-radius assumptions are largely compatible with the model, but their
diameters are comparable to the size of blood cells. As a result, microscopic
effects such as cell partitioning, near-wall interactions, and non-Newtonian
blood rheology may become important. These vessel classes are therefore less
compatible with the proposed model and require more
specialized microvascular transport models.

Overall, Table~\ref{tab:Vessel_Types} indicates that several vessel classes
are physiologically compatible with the proposed analytical model when it is
interpreted as a local vessel-segment model. The most suitable cases are
those in which the vessel geometry is locally simple and the wall-rigidity
and Newtonian-fluid assumptions are at least conditionally satisfied.
However, the final applicability of the \ac{1D} reduction
also depends on the specific \ac{Tx}-\ac{Rx} distance \(x_{\mathrm{Rx}}\) and
on the type of signaling molecule through its diffusion coefficient \(D\), as both
enter the admissibility criterion \(S_{\mathrm{Rx}}\ge2\). Thus, the proposed
model is expected to be reliable when both the vessel-class assumptions and
the admissibility criterion are satisfied for the considered
scenario. Weak pulsatility does not invalidate the formulation; it only reduces the expected difference between the
time-variant model and  steady-flow models. Conversely, stronger
pulsatility makes the time-variant formulation more important whenever the
remaining reduced-model assumptions are sufficiently satisfied.

\section{Conclusion}
\label{sec:conclusion}

In this paper, we developed a time-variant analytical model for dispersive
\ac{MC} channels with pulsatile flow. The model describes the local
straight-duct \ac{Rx} response through a Normal distribution with
time-variant mean and variance and extends this response to recirculating
closed-loop channels using a wrapped-Normal representation. Comparison with
\ac{3D} \acp{PBS} for synthetic and physiologically motivated velocity
waveforms showed close agreement over the considered parameter regimes.

The analysis provides several general insights into pulsatile molecular
transport. First, pulsatility has the strongest impact on the received signal
when transport is predominantly advective and occurs over time scales for
which the instantaneous flow variations are not averaged out. In contrast,
stronger diffusive spreading, longer propagation times, and faster flow
oscillations reduce the sensitivity of the received signal to the
instantaneous pulsatile waveform. Hence, a steady-flow approximation becomes
increasingly appropriate when temporal averaging dominates, whereas a
time-variant description is important when pulsatile effects remain
pronounced.

Second, the applicability analysis shows that the accuracy of the reduced
\ac{1D} model can be related to the nondimensional indicator
\(S_{\mathrm{Rx}}\). The \ac{PBS}-based parameter sweep supports
\(S_{\mathrm{Rx}}\ge2\) as a practical criterion for the applicability
of the proposed model. However, for blood vessels, this criterion alone is not sufficient
to establish physiological applicability. Our vessel-class analysis shows
that the model is most appropriate for locally uniform vessel segments for
which the rigid-wall and Newtonian-fluid assumptions are sufficiently
satisfied, whereas large compliant arteries and microvascular vessels may
require additional physiological modeling.

Finally, the proposed first-order peak-time approximation provides a
computationally efficient way to estimate pulsatility-induced shifts of the
first-arrival peak without having to numerically determine the maximum of the full \ac{Rx} response.
Together, these results provide guidance on when pulsatile flow needs to be
modeled explicitly and when the proposed reduced analytical description can
be expected to be accurate.

Future work includes extending the framework to branched vascular networks
and spatially varying vessel properties, incorporating effects such as vessel
wall compliance, and validating the model experimentally under
physiologically representative flow conditions.

\section*{Acknowledgments}
This work was funded in part by the German Federal Ministry of Research, Technology and Space (BMFTR) through Project Internet of Bio-Nano-Things (IoBNT) -- grant number 16KIS1987, in part by the German Research Foundation (Deutsche Forschungsgemeinschaft, DFG) under GRK 2950 -- ProjectID 509922606, grant number SCHA 2350/2-1 and SCHA 2350/4-1, and in part by the European Union’s Horizon Europe -- HORIZON-EIC-2024-PATHFINDEROPEN-01 under grant agreement Project N. 101185661.

\appendices

\section{Proof of Lemma~\ref{lem:straight_duct}}
\label{app:straight_duct_proof}

\begin{proof}
To obtain solution \eqref{eq:gaussian_infinite} for $p(x,t)$, we transform
\eqref{eq:adv_diff} into a diffusion equation with constant coefficients. We
first introduce the moving spatial coordinate
\begin{equation}
\xi = x - \mu(t), \qquad \mu(t) = \int_0^t u(s)\,\mathrm{d}s,
\end{equation}
where $\mu(t)$ represents the cumulative axial displacement of particles by
the time-variant flow. Next, we define the transformed concentration
\begin{equation}
\tilde p(\xi,t) := p(\xi+\mu(t),t),
\end{equation}
and since $\frac{\mathrm{d}\mu(t)}{\mathrm{d}t}=u(t)$, the chain rule yields
\begin{align}
\partial_x p(x,t) &= \partial_\xi \tilde p(\xi,t),\\
\partial_{xx} p(x,t) &= \partial_{\xi\xi} \tilde p(\xi,t),\\
\partial_t p(x,t) &= \partial_t \tilde p(\xi,t) - u(t)\partial_\xi \tilde p(\xi,t).
\end{align}
Substituting these expressions into \eqref{eq:adv_diff}, the advection terms
cancel and we obtain a diffusion equation with time-variant diffusion
coefficient
\begin{equation}
\partial_t \tilde p(\xi,t)=D_{\mathrm{1D}}(t)\partial_{\xi\xi}\tilde p(\xi,t).
\label{eq:transf}
\end{equation}
Next, we introduce a transformed time variable
\begin{equation}
\tau(t)=\int_0^t D_{\mathrm{1D}}(s)\,\mathrm{d}s,
\end{equation}
which normalizes the time axis such that the diffusion coefficient becomes
constant, i.e., $\mathrm{d}\tau/\mathrm{d}t=D_{\mathrm{1D}}(t)$. Let
$h(\xi,\tau(t)) = \tilde p(\xi,t)$. Differentiating with respect to $t$ yields
\begin{equation}
\partial_t \tilde p(\xi,t)=D_{\mathrm{1D}}(t)\partial_\tau h(\xi,\tau).
\end{equation}
Substituting into \eqref{eq:transf} leads to a diffusion equation with
constant coefficients
\begin{equation}
\partial_\tau h(\xi,\tau)=\partial_{\xi\xi}h(\xi,\tau).
\label{eq:trans-diff}
\end{equation}
For an impulsive release of particles at $t=0$ and $x=0$, the transformed
initial condition is $h(\xi,0)=\delta(\xi)$, where $\delta(\cdot)$ denotes the Dirac delta function. Solving \eqref{eq:trans-diff}
in an infinitely long straight duct yields, after returning to the original
variables~\cite{Jamali2019ChannelModeling}, the solution
$p_\mathrm{sd}(x,t)$ for \eqref{eq:adv_diff} in \eqref{eq:gaussian_infinite}.
\end{proof}

\section{Proof of Theorem~\ref{thm:closed_loop}}
\label{app:closed_loop_proof}

\begin{proof}
To obtain $p(x,t)$ for closed-loop systems \eqref{eq:cir}, we first restrict
the axial coordinate $x$ to the interval $[0,L]$ and impose periodic boundary
conditions, i.e.,
\begin{equation}
p(0,t)=p(L,t), \qquad \partial_x p(0,t)=\partial_x p(L,t).
\label{eq:periodic_bc}
\end{equation}
Then, the corresponding closed-loop solution \eqref{eq:cir} can be obtained
in two ways. First, one may solve \eqref{eq:adv_diff} directly on $[0,L]$
following the approach described for the straight-duct case in
Section~\ref{subsec:str_duct}, subject to the periodic boundary conditions in
\eqref{eq:periodic_bc}. Second, one may start directly from
\eqref{eq:gaussian_infinite} and construct a periodic representation by
summing over all integer shifts of length $L$, which effectively wraps the
straight-duct solution \eqref{eq:gaussian_infinite} around a circle of
circumference $L$, in line with the approach in~\cite{vakilipoor2025cam}.
\end{proof}

\section{Derivation of the First-Peak Approximation}
\label{app:first_peak_derivation}

For the local-concentration approximation in
\eqref{eq:pointwise_direct_concentration}, we define
\begin{equation}
a(t)=x_{\mathrm{Rx}}-\mu(t),
\qquad
s(t)=\sigma^2(t).
\end{equation}
The logarithm of the first-arrival local concentration is, up to an additive
constant,
\begin{equation}
\log p_0(x_{\mathrm{Rx}},t)
=
-\frac{1}{2}\log s(t)
-
\frac{a^2(t)}{2s(t)}.
\end{equation}
Using \(a'(t)=-u(t)\), where the prime denotes differentiation with respect
to \(t\), the peak condition
\(\frac{\mathrm{d}}{\mathrm{d}t}\log p_0(x_{\mathrm{Rx}},t)=0\) can be written as
\begin{equation}
F(t)
=
a(t)u(t)s(t)
+
\frac{1}{2}
\left[
a^2(t)-s(t)
\right]
s'(t)
=0.
\label{eq:appendix_peak_condition}
\end{equation}
For the steady-flow case,
\begin{equation}
u(t)=\bar u,
\qquad
a_{\mathrm{c}}(t)=x_{\mathrm{Rx}}-\bar u t,
\qquad
s_{\mathrm{c}}(t)=2D_{\mathrm{eff}}^{\mathrm{const}}t.
\end{equation}
Substituting these expressions into \eqref{eq:appendix_peak_condition} yields
the steady-flow peak condition. Solving it gives the peak time \(t_0\) in
\eqref{eq:t0_first_peak}. We denote the corresponding steady-flow peak
function by \(F_0(t)\). Differentiating \(F_0(t)\) and evaluating the result
at \(t_0\) yields \eqref{eq:F0prime_peak}.

The derivative \(F_0'(t_0)\) determines the local sensitivity of the
steady-flow peak condition and will enter the first-order correction of the
peak time. We now introduce the pulsatile modulation as a perturbation of
the steady-flow system.

\begin{equation}
u(t)=\bar u[1+q(t)],
\quad
q(t)=
\sum_{n=1}^{N}
M_n\cos(\omega_n t+\phi_n),
\quad
\omega_n=n\omega.
\end{equation}
The accumulated modulation is
\begin{equation}
Q(t)
=
\int_0^t q(s)\,\mathrm{d}s
=
\sum_{n=1}^{N}
\frac{M_n}{\omega_n}
\left[
\sin(\omega_n t+\phi_n)-\sin(\phi_n)
\right].
\end{equation}
Retaining only first-order terms in the modulation amplitudes \(M_n\), the
mean displacement and variance become
\begin{align}
\mu(t)
&=
\bar u t+\bar u Q(t)
+
\mathcal{O}(\|\mathbf{M}\|^2),
\\
\sigma^2(t)
&=
2D_{\mathrm{eff}}^{\mathrm{const}}t
+
4\kappa \bar u^2 Q(t)
+
\mathcal{O}(\|\mathbf{M}\|^2),
\end{align}
where \(\kappa=R^2/(48D)\), \(\mathbf{M}=[M_1,\ldots,M_N]^{\mathsf T}\), and
\(\mathcal{O}(\|\mathbf{M}\|^2)\) collects terms of second and higher order
in the modulation amplitudes. Hence,
\begin{align}
a(t)
&=
a_{\mathrm{c}}(t)-\bar u Q(t)
+
\mathcal{O}(\|\mathbf{M}\|^2),
\\
s(t)
&=
s_{\mathrm{c}}(t)+4\kappa \bar u^2 Q(t)
+
\mathcal{O}(\|\mathbf{M}\|^2),
\\
s'(t)
&=
2D_{\mathrm{eff}}^{\mathrm{const}}
+
4\kappa \bar u^2 q(t)
+
\mathcal{O}(\|\mathbf{M}\|^2).
\end{align}

Substituting these first-order expansions into
\eqref{eq:appendix_peak_condition} gives
\begin{equation}
F(t)
=
F_0(t)+F_1(t)+\mathcal{O}(\|\mathbf{M}\|^2),
\end{equation}
where \(F_1(t)\) contains the terms that are linear in the pulsatile
modulation. Since \(F_0(t_0)=0\), \(F_1(t_0)\) represents the first-order
perturbation of the peak condition caused by the pulsatile waveform.

Evaluating \(F_1(t)\) at \(t=t_0\) and using the definitions in
\eqref{eq:peak_auxiliary_defs} and \eqref{eq:q0_Q0_defs} yields
\eqref{eq:F1_peak}.

Finally, let the pulsatile peak time be written as
\begin{equation}
t_{\mathrm{peak}}=t_0+\Delta t,
\end{equation}
where \(\Delta t=\mathcal{O}(\|\mathbf{M}\|)\). A first-order Taylor expansion of
\(F(t_{\mathrm{peak}})=0\) around \(t_0\) gives
\begin{equation}
0
=
F_0(t_0)
+
F_0'(t_0)\Delta t
+
F_1(t_0)
+
\mathcal{O}(\|\mathbf{M}\|^2).
\end{equation}
Since \(F_0(t_0)=0\), we obtain
\begin{equation}
\Delta t
=
-
\frac{F_1(t_0)}{F_0'(t_0)}.
\end{equation}
Here, \(F_0'(t_0)\) quantifies the local sensitivity of the steady-flow peak
condition around \(t_0\). Therefore, the ratio
\(-F_1(t_0)/F_0'(t_0)\) converts the pulsatile perturbation of the peak
condition into a first-order shift of the peak time, which directly yields
\eqref{eq:tpeak_first_order}.

\section*{Acknowledgments}
This work was funded in part by the German Federal Ministry of Research, Technology and Space (BMFTR) through Project Internet of Bio-Nano-Things (IoBNT) -- grant number 16KIS1987, in part by the German Research Foundation (Deutsche Forschungsgemeinschaft, DFG) under GRK 2950 -- ProjectID 509922606, grant number SCHA 2350/2-1 and SCHA 2350/4-1, and in part by the European Union’s Horizon Europe -- HORIZON-EIC-2024-PATHFINDEROPEN-01 under grant agreement Project N. 101185661.

\ifCLASSOPTIONcaptionsoff
  \newpage
\fi

\bibliographystyle{IEEEtran}
\bibliography{references}

\end{document}

%% file: acronyms.tex
\acrodef{MC}{molecular communication}
\acrodef{CIR}{channel impulse response}
\acrodef{PBS}{particle-based simulation}
\acrodef{1D}{one-dimensional}
\acrodef {3D}{three-dimensional}
\acrodef{PDF}{probability density}
\acrodef{CVS}{cardiovascular system}
\acrodef{Tx}{transmitter}
\acrodef{Rx}{receiver}